\documentclass[USenglish]{lipics-v2016}

\newcommand{\lv}[1]{#1}
\newcommand{\sv}[1]{}

\usepackage[utf8]{inputenc}
\usepackage{mathscinet}
\usepackage{comment}
\usepackage{graphicx}
\usepackage{booktabs}
\usepackage{hyperref}
\usepackage{pgfkeys}
\usepackage[numbers,sort&compress]{natbib}
\usepackage{enumerate}
\usepackage{verbatim}
\usepackage{nicefrac}
\usepackage{framed}

\usepackage[dvipsnames]{xcolor}
\usepackage{tikz}
\usetikzlibrary{fit}
\usetikzlibrary{intersections}
\usetikzlibrary{decorations.pathreplacing}
\usetikzlibrary{calc}

\usepackage{amsmath}
\usepackage{amssymb}
\usepackage{thmtools}
\usepackage{thm-restate}

\usepackage{xspace}
\usepackage{authblk}
\usepackage[color=green!20]{todonotes}

\usepackage[capitalise]{cleveref}

\newcommand{\rbs}{\textsc{Red-Blue Separation}\xspace}
\newcommand{\hvrbs}{\textsc{Axis-Parallel Red-Blue Separation}\xspace}

\def\shs{\textsc{Structured 2-Track Hitting Set}\xspace}
\def\shss{\textsc{S2-THS}\xspace}

\def\wone{$W[1]$\xspace}

\tikzstyle{rp}=[circle,fill=red,inner sep=-0.05cm]
\tikzstyle{bp}=[circle,fill=blue,inner sep=-0.05cm]

\tikzstyle{rp2}=[circle,fill=red,inner sep=-0.03cm]
\tikzstyle{bp2}=[circle,fill=blue,inner sep=-0.03cm]

\title{On the Parameterized Complexity of Red-Blue Points Separation\footnote{Research partially supported by EPSRC grant EP/N029143/1.}
}

\author[1]{\'Edouard Bonnet}
\author[1]{Panos Giannopoulos}
\author[2]{Michael Lampis}

\affil[1]{Middlesex University, Department of Computer Science, London, UK \texttt{edouard.bonnet@dauphine.fr}, \texttt{p.giannopoulos@mdx.ac.uk}}
\affil[2]{Universit\'{e} Paris-Dauphine, PSL Research University, CNRS, LAMSADE, Paris, France \texttt{michail.lampis@dauphine.fr}}		
\authorrunning{\'E. Bonnet, P. Giannopoulos, M. Lampis}
\Copyright{\'Edouard Bonnet, Panos Giannopoulos, Michael Lampis}
\subjclass{F.2.2 Nonnumerical Algorithms and Problems}
\keywords{red-blue points separation, geometric problem, W[1]-hardness, FPT algorithm, ETH-based lower bound}

\begin{document}

\maketitle

\begin{abstract}
We study the following geometric separation problem: 
  Given a set $\mathcal R$ of red points and a set $\mathcal B$ of blue points in the plane, find a minimum-size set of lines that separate $\mathcal R$ from $\mathcal B$.
  We show that, in its full generality, parameterized by the number of lines $k$ in the solution, the problem is unlikely to be solvable significantly faster than the brute-force $n^{O(k)}$-time algorithm, where $n$ is the total number of points.
  Indeed, we show that an algorithm running in time $f(k)n^{o(k/ \log k)}$, for any computable function $f$, would disprove ETH.
  Our reduction crucially relies on selecting lines from a set with a large number of different slopes (i.e., this number is not a function of $k$).

  Conjecturing that the problem variant where the lines are required to be axis-parallel is FPT in the number of lines, we show the following preliminary result. 
Separating  $\mathcal R$ from $\mathcal B$ with a minimum-size set of axis-parallel lines is FPT in the size of either set, and can be solved in time $O^*(9^{|\mathcal B|})$ (assuming that $\mathcal B$ is the smallest set).
\end{abstract}

\section{Introduction}\label{sec:intro}

We study the parameterized complexity of the following \rbs problem: 
Given a set $\mathcal R$ of red points and a set $\mathcal B$ of blue points in the plane and a positive integer $k$, find a set of at most $k$ lines that together separate $\mathcal R$ from $\mathcal B$ (or report
that such a set does not exist). Separation here means that each cell in the arrangement
induced by the lines in the solution is either monochromatic, i.e., contains points of one color only, or empty. Equivalently, $\mathcal R$ is separated from $\mathcal B$ if every straight-line segment with one endpoint in $\mathcal R$ and the other one in $\mathcal B$ is intersected by at least one line in the solution. Note here that we opt for \emph{strict} separation that is, no point in $\mathcal R \cup \mathcal B$ is on a separating line. Let $n := |\mathcal R \cup \mathcal B| $.

The variant where the separating lines sought must be axis-parallel will be simply referred to as \hvrbs.

Apart from being interesting in its own right, \rbs is also directly motivated by the problem of
univariate discretization of continuous variables in the context of machine learning \cite{FI93, KE07}. 
For example, its two-dimensional
version models problem instances with decision tables of two real-valued attributes and a binary
decision function. The lines to be found represent cut points determining a partition of the
values into intervals and one opts for a minimum-size set of cuts that is consistent with the given decision table.
The problem is also known as \emph{minimum linear classification}; see \cite{LDJXZ06} for an application in signal processing.
For the case where $k=1$ and $k=2$, \rbs is solvable in $O(n)$ and $O(n\log n)$ time respectively \cite{HMRS04}.
When $k$ is part of the input, it is known to be NP-hard \cite{Meg88} and APX-hard \cite{CDKW05} even for the axis-parallel variant. The latter also admits a $2$-approximation \cite{CDKW05}.

\medskip
\noindent
\textbf{Results.}  We first show that \rbs is W[1]-hard in the solution size $k$ and that it cannot be solved in $f(k)n^{o(k/ \log k)}$ time (for any computable function $f$) unless ETH fails. Our reduction is from \shs, see Section~\ref{sec:hardness}, which has been recently used for showing hardness for another classical geometric optimization problem \cite{BonnetM16}.
Then, in Section~\ref{sec:fpt}, we show that \hvrbs is FPT in the size of either of $\mathcal R$ and $\mathcal B$. Our algorithm is simple and is based on reducing the problem to $9^{|\mathcal B|+2}$ instances of \textsc{2-SAT} (assuming, w.l.o.g., that $\mathcal B$ is the smallest set).  

\medskip
\noindent
\textbf{Related work.}
The following monochromatic points separation problem has also been studied:  Given a set of points in the plane, find a smallest set of lines that separates every point from every other point in the set (i.e., each cell in the induced arrangement must contain at most one point).
It has been shown to be NP-hard \cite{FMP91}, APX-hard \cite{CDKW05} and, in the axis-parallel case, to admit a $2$-approximation \cite{CDKW05}.
Very recently, the problem has been also shown to admit an $\text{OPT} \log \text{OPT}$-approximation \cite{HPJ17}.
Note here that it is trivially FPT in the number of lines, as the number of cells in the arrangement of $k$ lines is at most $\Theta(k^2)$. For results on several other related separation problems, see \cite{HMRS04, DHMS01}.



\lv{
  
\lv{
\section{Preliminaries}
}
\lv{
For positive integers $x$, $y$, let $[x]$ be the set of integers between 1 and $x$, and $[x,y]$ the set of integers between $x$ and $y$.

}

\sv{For positive integers $x$, $y$, let $[x]$ be the set of integers between 1 and $x$, and $[x,y]$ the set of integers between $x$ and $y$.
}
For a totally ordered (finite) set $X$, an \emph{$X$-interval} is any subset of $X$ of consecutive elements. 
In the \textsc{2-Track Hitting Set} problem, the input consists of an integer $k$, two totally ordered ground sets $A$ and $B$ of the same cardinality, and two sets $\mathcal S_A$ of $A$-intervals and $\mathcal S_B$ of $B$-intervals.
The elements of $A$ and $B$ are in one-to-one correspondence $\phi: A \rightarrow B$ and each pair $(a,\phi(a))$ is called a \emph{$2$-element}. 
The goal is to decide if there is a set $S$ of $k$ $2$-elements such that the first projection of $S$ is a hitting set of $\mathcal S_A$, and the second projection of $S$ is a hitting set of $\mathcal S_B$.
We will refer to the interval systems $(A,\mathcal S_A)$ and $(B,\mathcal S_B)$ as track A and track B.

\shs (\shss for short) is the same problem with color classes over the $2$-elements and a restriction on the one-to-one mapping $\phi$; see Figure~\ref{fig:permutations} for an illustration. 
Given two integers $k$ and $t$, $A$ is partitioned into $(C_1,C_2,\ldots,C_k)$ where $C_j=\{a^j_1, a^j_2, \ldots, a^j_t\}$ for each $j \in [k]$.
$A$ is ordered: $a^1_1, a^1_2, \ldots, a^1_t, a^2_1, a^2_2, \ldots, a^2_t, \ldots, a^k_1, a^k_2, \ldots, a^k_t$.
We define $C'_j:=\phi(C_j)$ and $b^j_i := \phi(a^j_i)$ for all $i \in [t]$ and $j \in [k]$.
We now impose that $\phi$ is such that, for each $j \in [k]$, the set $C'_j$ is a $B$-interval.
That is, $B$ is ordered: $C'_{\sigma(1)}, C'_{\sigma(2)}, \ldots, C'_{\sigma(k)}$ for some permutation on $[k]$, $\sigma \in \mathfrak S_k$.
For each $j \in [k]$, the order of the elements within $C'_j$ can be described by a permutation $\sigma_j \in \mathfrak S_t$ such that the ordering of $C'_j$ is: $b^j_{\sigma_j(1)}, b^j_{\sigma_j(2)}, \ldots, b^j_{\sigma_j(t)}$.
In what follows, it will be convenient to see an instance of \shss as a tuple $\mathcal I=(k \in \mathbb{N},t \in \mathbb{N}, \sigma \in \mathfrak S_k, \sigma_1 \in \mathfrak S_t, \ldots, \sigma_k \in \mathfrak S_t, \mathcal S_A, \mathcal S_B)$, where $\mathcal S_A$ is a set of $A$-intervals and $\mathcal S_B$ is a set of $B$-intervals.
We denote by $[a^j_i,a^{j'}_{i'}]$ (resp. $[b^j_i,b^{j'}_{i'}]$) all the elements $a \in A$ (resp. $b \in B$) such that $a^j_i \leq_A a \leq_A a^{j'}_{i'}$ (resp. $b^j_i \leq_B b \leq_B b^{j'}_{i'}$).

\begin{figure}[htbp]
\centering
\begin{tikzpicture}

\definecolor{col1}{rgb}{1,0,0}
\definecolor{col2}{rgb}{0,1,0}
\definecolor{col3}{rgb}{0,0,1}
\definecolor{col4}{rgb}{1,1,0}
\def\p{-2}
\def\e{0.5}

\def\t{6}
\def\k{4}
\foreach \j [count = \ja from 0] in {1,...,\k}{
\foreach \i in {1,...,\t}{
\node (a\j\i) at (\e * \i + \e * \ja * \t, 0) {$a^\j_\i$} ;
}
\node[draw, fill=col\j, opacity=0.3, rectangle, rounded corners, thick, fit=(a\j1) (a\j\t),label=above:$C_\j$,inner sep=-0.05cm] (c\j) {} ;
}
\node[rectangle, rounded corners, thick, fit=(c1) (c\k),label=north:$A$,inner sep=0.4cm] (A) {} ;

\node (b34) at (\e, \p) {$b^3_4$} ;
\node (b32) at (2*\e, \p) {$b^3_2$} ;
\node (b33) at (3*\e, \p) {$b^3_3$} ;
\node (b36) at (4*\e, \p) {$b^3_6$} ;
\node (b31) at (5*\e, \p) {$b^3_1$} ;
\node (b35) at (6*\e, \p) {$b^3_5$} ;
\node[draw, fill=col3, opacity=0.3, rectangle, rounded corners, thick, fit=(b34) (b35),label=below:$C'_3$,inner sep=-0.05cm] (cp3) {} ;

\begin{scope}[xshift=6*\e cm]
\node (b12) at (\e, \p) {$b^1_2$} ;
\node (b14) at (2*\e, \p) {$b^1_4$} ;
\node (b11) at (3*\e, \p) {$b^1_1$} ;
\node (b15) at (4*\e, \p) {$b^1_5$} ;
\node (b16) at (5*\e, \p) {$b^1_6$} ;
\node (b13) at (6*\e, \p) {$b^1_3$} ;
\node[draw, fill=col1, opacity=0.3, rectangle, rounded corners, thick, fit=(b12) (b13),label=below:$C'_1$,inner sep=-0.05cm] (cp1) {} ;
\end{scope}

\begin{scope}[xshift=12*\e cm]
\node (b43) at (\e, \p) {$b^4_3$} ;
\node (b46) at (2*\e, \p) {$b^4_6$} ;
\node (b45) at (3*\e, \p) {$b^4_5$} ;
\node (b42) at (4*\e, \p) {$b^4_2$} ;
\node (b41) at (5*\e, \p) {$b^4_1$} ;
\node (b44) at (6*\e, \p) {$b^1_4$} ;
\node[draw, fill=col4, opacity=0.3, rectangle, rounded corners, thick, fit=(b43) (b44),label=below:$C'_4$,inner sep=-0.05cm] (cp4) {} ;
\end{scope}

\begin{scope}[xshift=18*\e cm]
\node (b21) at (\e, \p) {$b^2_1$} ;
\node (b25) at (2*\e, \p) {$b^2_5$} ;
\node (b22) at (3*\e, \p) {$b^2_2$} ;
\node (b24) at (4*\e, \p) {$b^2_4$} ;
\node (b26) at (5*\e, \p) {$b^2_6$} ;
\node (b23) at (6*\e, \p) {$b^2_3$} ;
\node[draw, fill=col2, opacity=0.3, rectangle, rounded corners, thick, fit=(b21) (b23),label=below:$C'_2$,inner sep=-0.05cm] (cp2) {} ;
\end{scope}

\node[rectangle, rounded corners, thick, fit=(cp3) (cp2),label=south:$B$,inner sep=0.4cm] (B) {} ;

\foreach \j in {1,...,\k}{
\draw[very thick,opacity=0.8] (c\j.south) -- (cp\j.north) ; 
}

\foreach \i in {1,...,\t}{
\draw (a1\i.south) -- (b1\i.north) ; 
}

\node (s) at (6.25,-1) {$\sigma$} ;
\node (s1) at (1.5,-1) {$\sigma_1$} ;

\node (oa) at (-0.1,0) {$\leq_A:$} ;
\node (ob) at (-0.1,-2) {$\leq_B:$} ;

\end{tikzpicture}
\caption{An illustration of the $k+1$ permutations $\sigma \in \mathfrak S_k$, $\sigma_1 \in \mathfrak S_t$, \dots, $\sigma_k \in \mathfrak S_t$ of an instance of \shs, with $k=4$ and $t=6$.}
\label{fig:permutations}
\end{figure}
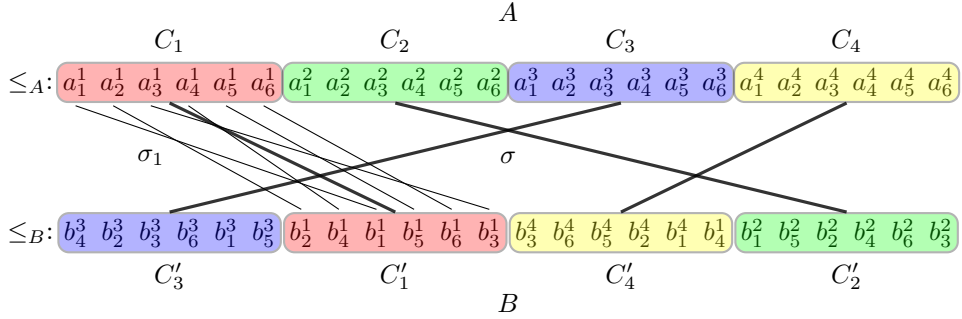

\lv{
\lv{\medskip}
\noindent
\textbf{ETH-based lower bounds.}
The \emph{Exponential Time Hypothesis} (ETH) is a conjecture by Impagliazzo et al.~\cite{impagliazzo1999complexity} asserting that there is no $2^{o(n)}$-time algorithm for \textsc{3-SAT} on instances with $n$ variables.

The \textsc{Multicolored Subgraph Isomorphism} problem can be defined in the following way.
One is given a graph with $n$ vertices partitioned into $l$ color classes $V_1, \ldots, V_l$ such that only $k$ of the ${l \choose 2}$ sets $E_{ij}=E(V_i,V_j)$ are non empty.
The goal is to pick one vertex in each color class so that the selected vertices induce $k$ edges. 
Observe that $l$ corresponds to the number of vertices of the pattern graph.
The technique of color coding and a result by Marx imply that:
\begin{theorem}[\cite{Marx10}]\label{thm:marx10}
\textsc{Multicolored Subgraph Isomorphism} cannot be solved in time $f(k)\, n^{o(k / \log k)}$ where $k$ is the number of edges of the solution and $f$ any computable function, unless the ETH fails.
\end{theorem}

Bonnet and Miltzow showed the following conditional lower bound for \shs by a reduction from \textsc{Multicolored Subgraph Isomorphism} linearly preserving the parameter:

\begin{theorem}[\cite{BonnetM16}]\label{thm:shs-lowerbound}
   \shs is \wone-hard and, unless the ETH fails, cannot be solved in time $f(k)n^{o(k/ \log k)}$ for any computable function $f$.
\end{theorem}

The same lower bound has been shown for \textsc{2-Track Hitting Set} by Marx and Pilipczuk \cite{MarxP15}.
They use this intermediate result to show that covering a given set of points in the plane with $k$ axis-parallel rectangles taken from a prescribed set cannot be solved in time $f(k)n^{o(k/ \log k)}$, even if the rectangles are almost squares.
Bonnet and Miltzow used Theorem~\ref{thm:shs-lowerbound} to show the same lower bound for \textsc{Point Guard Art Gallery} and \textsc{Vertex Guard Art Gallery}, where one wants to guard a simple polygon with $k$ points, and $k$ vertices, respectively.
In this paper, we again utilize \shss for a reduction to \rbs.
Thus, it seems as though \textsc{(Structured) 2-Track Hitting Set} can be a good starting point for a wide variety of geometric problems and yield almost tight lower bounds, like \textsc{Grid Tiling} \cite{Marx06} has been doing in the last decade for geometric problems optimally solvable in $n^{\Theta(\sqrt k)}$.

}

}
\section{Parameterized hardness for arbitrary slopes} 
\label{sec:hardness}

In this section, we show that \rbs is unlikely to be FPT with respect to the number of lines $k$ and establish that, unless the ETH fails, the $n^{O(k)}$-time brute-force algorithm is almost optimal.
\sv{
  We reduce from \shs \cite{BonnetM16}, see below; more about this problem and its lower bound can be found in the appendix.


  }

\lv{
Let us say a few words about the difficulty of showing such a result for \rbs, compared to its NP-hardness.
A set of $k$ lines creates at most $h(k) := {k+1 \choose 2}+1$ cells.
Therefore, any YES-instance can be covered by $h(k)$ pairwise-disjoint monochromatic convex sets.
This prevents us from encoding an adjacency matrix on $n$ vertices with bichromatic gadgets, while one does not seem to achieve much with a monochromatic encoding.

A perhaps more concrete issue with encoding an adjacency matrix is the following.
Suppose we try to reduce directly from \textsc{Multicolored Subgraph Isomorphism} (or its special case \textsc{Multicolored Clique}), and we want a horizontal line $L(u)$ to represent the choice of a vertex $u$ within one set, a vertical line $L(v)$ to represent the choice of a vertex $v$ in another set, and the lines are compatible iff $uv$ is an edge.
Here is the pitfall: if $uv$ and $uw$ are edges, then $L(u)$ should form a feasible solution with $L(v)$ and with $L(w)$; but then, it can be observed that every vertical line in between $L(v)$ and $L(w)$ also completes $L(u)$ into a feasible solution (which is undesirable as soon as there are vertices between $v$ and $w$ which are not adjacent to $u$).

We overcome those issues by reducing from \shs.

}
If one deconstructs \shss, one finds intervals, a permutation of the color classes $\sigma$, and $k$ permutations $\sigma_j$'s of the elements within the classes.
Intervals, thanks to their geometric nature, can be realized by two red points which have to be separated from a diagonal of blue points (see Figure~\ref{fig:interval}), while permutation $\sigma$, being on $k$ elements, can be designed straightforwardly without blowing-up the size of the solution (see Figure~\ref{fig:permutation-classes}).
For these gadgets, we would like to force the chosen lines to be axis-parallel.
We obtain this by surrounding them with \emph{long alleys} made off long red paths parallel and next to long blue paths (see Figure~\ref{fig:long-alley-gadget}).
The main challenge is to get the permutations $\sigma_j$'s on $t$ elements.
To attain this, we match a selected line $L_i$ (corresponding to an element of index $i \in [t]$) to a specific angle $\alpha_i$, which \emph{leads} to the intended position of the element of index $\sigma_j(l) = i$, for some $l \in [t]$ (see Figure~\ref{fig:permutation-within-a-class}).
Note that the depicted gadget actually links the element of index $i$ to elements equal to or smaller than the element indexed at $\sigma_j(l)$. By combining two of these gadgets we can easily obtain only the intended position (see Figure~\ref{fig:overall}).


\begin{theorem}
  \rbs is \wone-hard w.r.t. the number of lines $k$, and unless ETH fails, cannot be solved in time $f(k)n^{o(k/ \log k)}$ for any computable function $f$. 
\end{theorem}

\begin{proof}
  We reduce from \shss, which is \wone-hard and has the above lower bound under ETH \cite{BonnetM16}.
  Let $\mathcal I=(k \in \mathbb{N},t \in \mathbb{N}, \sigma \in \mathfrak S_k, \sigma_1 \in \mathfrak S_t, \ldots, \sigma_k \in \mathfrak S_t, \mathcal S_A, \mathcal S_B)$ be an instance of \shss.
  We will build an instance $\mathcal J=(\mathcal R,\mathcal B, 6k+14)$ of \rbs such that $\mathcal I$ is a YES-instance for \shss if and only if $\mathcal R$ and $\mathcal B$ can be separated with $6k+14$ lines.

  The points in $\mathcal R$ and $\mathcal B$ will have rational coordinates.
  More precisely, most points will be pinned to a $z$-by-$z$ grid where $z$ is polynomial in the size of $\mathcal I$.
  The rest will have rational coordinates with nominator and denominator polynomial in $z$. 
 Let $\Gamma$ be the $z$-by-$z$ grid corresponding to the set of points with coordinates in $[z] \times [z]$.
  We call \emph{horizontally} (resp.~vertically) \emph{consecutive points} a set of points of $\Gamma$ with coordinates $(a,y), (a+1,y), \ldots (b-1,y), (b,y)$ for $a,b,y \in [z]$ and $a<b$ (resp. $(x,a), (x,a+1), \ldots (x,b-1), (x,b)$ for $a,b,x \in [z]$ and $a<b$).
  We denote those points by $\mathcal C(a \rightarrow b,y)$ (resp. $\mathcal C(x,a \rightarrow b)$).
 
 \medskip
\noindent  
  \textbf{Long alley gadgets.}
  In the gadgets encoding the intervals (see next paragraph), we will need to restrict the selected separating lines to be almost horizontal or almost vertical.
To enforce that, we use the \emph{long alley} gadgets.
A \emph{horizontal long alley} gadget is made of $\ell$ horizontally consecutive red points $\mathcal C(a \rightarrow a+\ell-1,y)$ and $\ell$ horizontally consecutive blue points $\mathcal C(a \rightarrow a+\ell-1,y')$ with $a,a+\ell-1,y \neq y' \in [z]$ (see Figure~\ref{fig:long-alley}).
A \emph{vertical long alley} is defined analogously.
Long alleys are called so because 
$\ell \gg |y-y'|$ thus, separating the red points from the blue points of a horizontal (resp. vertical) long alley with a budget of only one line, requires the line to be almost horizontal (resp. vertical).
The use of the long alleys will be the following.
Let $\mathcal G$ be a gadget for which we wish the separating lines to be almost horizontal or vertical.
Say, $\mathcal G$ occupies a $g$-by-$g$ subgrid of $\Gamma$ (with $g \ll z$).
We place four long alley gadgets to the left, top, right, and bottom of $\mathcal G$: horizontal ones to the left and right, vertical ones to the top and bottom (as depicted in Figure~\ref{fig:long-alley-usage}).
The left horizontal (resp. bottom vertical) long alley starts at the $x$-coordinate (resp. $y$ coordinate) of $1$, whereas the right horizontal (resp. top vertical) long alley ends at the $x$-coordinate (resp. $y$ coordinate) of $z$; see Figure~\ref{fig:overall}, where the long alleys are depicted by thin rectangles.

Note that we will not surround each and every gadget of the construction by four long alleys.
At some places, it will indeed be crucial that the lines can have arbitrary slopes.

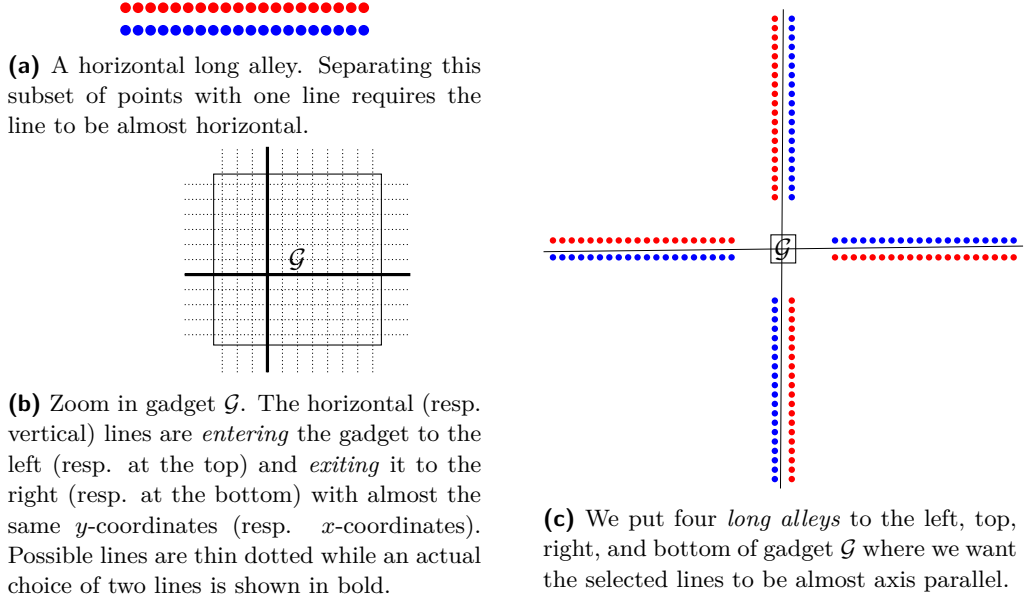
\begin{figure}
  \centering
   \begin{minipage}{0.45\textwidth}
     \centering
     \begin{tikzpicture}
       \def\s{6}
       \def\d{0.5}

       \foreach \i in {1,...,20}{
         \node[bp] at (\i / \s,0) {} ;
         \node[rp] at (\i / \s,0.3) {} ;
       }
     \end{tikzpicture}
     \subcaption{A horizontal long alley. Separating this subset of points with one line requires the line to be almost horizontal.}
     \label{fig:long-alley}

     \qquad
     \qquad
     \centering
     \begin{tikzpicture}
       \def\t{5}
       \def\s{0.2}

       \foreach \i in {-\t,...,\t}{
         \draw[densely dotted, thin] (\i * \s,-1.5) -- (\i * \s,1.5);
         \draw[densely dotted, thin] (-1.5,\i * \s) -- (1.5,\i * \s);
       }

       \draw[very thick] (-2 * \s,-1.5) -- (-2 * \s,1.5);
       \draw[very thick] (-1.5,-1 * \s) -- (1.5,-1 * \s);
       \node[draw,rectangle,inner sep=1cm] at (0,0) {$\mathcal G$} ;
     \end{tikzpicture}
     \subcaption{Zoom in gadget $\mathcal G$.
       The horizontal (resp. vertical) lines are \emph{entering} the gadget to the left (resp. at the top) and \emph{exiting} it to the right (resp. at the bottom) with almost the same $y$-coordinates (resp. $x$-coordinates).
       Possible lines are thin dotted while an actual choice of two lines is shown in bold.}
     \label{fig:almost-axis-parallel}
    \end{minipage}
     \qquad
    \begin{minipage}{0.45\textwidth}
      \centering
      \begin{tikzpicture}[scale=0.75]
       \def\s{6}
       \def\d{0.5}
       \def\t{20}
       \foreach \i in {1,...,\t}{
         \node[bp2] at (\i / \s,0) {} ;
         \node[rp2] at (\i / \s,0.3) {} ;
       }
       \foreach \i in {1,...,\t}{
         \node[rp2] at (\i / \s + 5,0) {} ;
         \node[bp2] at (\i / \s + 5,0.3) {} ;
       }
       \foreach \i in {1,...,\t}{
         \node[rp2] at (4.1,\i / \s + 0.9) {} ;
         \node[bp2] at (4.4,\i / \s + 0.9) {} ;
       }
       \foreach \i in {1,...,\t}{
         \node[bp2] at (4.1,\i / \s - 4.1) {} ;
         \node[rp2] at (4.4,\i / \s - 4.1) {} ;
       }

       \draw[very thin] (0, 0.1) -- (51 / \s, 0.2) ;
       \draw[very thin] (4.2, -4.1) -- (4.25, \t / \s + 0.9 + 1 / \s) ;
       \node[draw,rectangle,inner sep=0.05cm] at (4.25,0.15) {$\mathcal G$} ;
     \end{tikzpicture}
      \subcaption{We put four \emph{long alleys} to the left, top, right, and bottom of gadget $\mathcal G$ where we want the selected lines to be almost axis parallel.}
      \label{fig:long-alley-usage}
    \end{minipage} 
\caption{The long alley gadget and its use in combination with another gadget.}
\label{fig:long-alley-gadget}
\end{figure}

\medskip
\noindent
\textbf{Interval gadgets and encoding track A.}
The elements of $A$ are represented by a diagonal of $kt-1$ blue points.
More precisely, we add the points $(x_0^{A},y_0^{A}), (x_0^{A}+4,y_0^{A}+4), (x_0^{A}+8,y_0^{A}+8), \ldots, (x_0^{A}+4kt-8,y_0^{A}+4kt-8)$ to $\mathcal B$ for some offset $x_0^{A},y_0^{A} \in [z]$ that we will specify later.
We think those points as going from the first $(x_0^{A},y_0^{A})$ to the last $(x_0^{A}+4kt-8,y_0^{A}+4kt-8)$.
An almost horizontal (resp.~vertical) line just below (resp.~just to the left of) the $s$-th blue point of this diagonal translates as selecting the $s$-th element of $A$ in the order fixed by $\leq_A$.
The almost horizontal (resp.~vertical) line just above (resp.~just to the right of) the last blue point corresponds to selecting the $kt$-th, i.e., last, element of $A$. 

For each interval $[a_i^j,a_{i'}^{j'}]$ in $\mathcal S_A$ (for some $i,i' \in [k]$, $j,j' \in [t]$), that is, the interval between the $s := ((j-1)t+i)$-th and the $s' := ((j'-1)t+i')$-th elements of $A$, we add two red points: one at $(x_0^{A}+4s-7,y_0^{A}+4s'-5)$ and one at $(x_0^{A}+4s'-5,y_0^{A}+4s-7)$ (see Figure~\ref{fig:interval-gadget} for one interval gadget and Figure~\ref{fig:interval-gadgets} for track $A$).
Let $R([a_i^j,a_{i'}^{j'}])$ be this pair of red points.
Informally, one red point has its projection along the $x$-axis just to the left of the $s$-th blue point and its projection along the $y$-axis just above the $s'$-th blue point; the other one has its projection along the $x$-axis just to the right of the $s'$-th blue point and its projection along the $y$-axis just below the $s$-th blue point.
For technical reasons, we add, for every $j \in [k]$, the pair $R([a_1^j,a_t^j])$ encoding the interval formed by all the elements of the $j$-th color class of $A$.
Adding these intervals to $\mathcal S_A$ does not constrain the problem more.

We surround this encoding of track $A$, which we denote by $\mathcal G(A)$, with $4k$ long alleys, 
whose width is $4t-4$, from $x$-coordinates $x_0^{A}+4(j-1)t-2$ to $x_0^{A}+4jt-6$ for vertical alleys (from $y$-coordinates $y_0^{A}+4(j-1)t-2$ to $y_0^{A}+4jt-6$ for horizontal alleys).
We alternate red-blue\footnote{i.e., for horizontal (resp.~vertical) alleys, the red points are above (resp.~to the left of) the blue points.} alleys and blue-red alleys for two contiguous alleys so that there is no need to separate one from the other.
We start with a red-blue alley for the left horizontal and top vertical groups of alleys, and with a blue-red alley for the right horizontal and bottom vertical.
This last detail is not in any way crucial but permits the construction to be defined uniquely and consistently with the choices of Figure~\ref{fig:long-alley-usage}.
This, together with the description of long alleys in the previous paragraph, fully defines the $4k$ long alleys (see Figure~\ref{fig:overall}).

The general intention is that in order to separate those two red points from the blue diagonal with a budget of two almost axis-parallel lines, one should take two lines (one almost horizontal and one almost vertical) corresponding to the selection of the same element of $A$ which hits the corresponding interval.
In particular, taking two almost horizontal lines (resp.~two almost vertical lines) is made impossible due to those vertical (resp.~horizontal) long alleys.
More precisely, the intended pairs of lines separating the red points $R([a_i^j,a_{i'}^{j'}])$ from the blue diagonal are of the form $x=x_0^{A}+4\hat s-6, y=y_0^{A}+4\hat s-6$ for $\hat s \in [s,s']$.
Furthermore, the $4k$ long alleys force a pair of (almost) horizontal and vertical lines corresponding to one element per color class to be taken.

For any $s \in [tk]$, $i \in [t]$, and $j \in [k]$, such that $s=(j-1)t+i$, let $\text{HL}(s)$ be the horizontal line of equation $y=y_0^{A}+4s-6$ and $\text{VL}(s)$ the vertical line of equation $x=x_0^{A}+4s-6$.
They \emph{correspond} to selecting $a_i^j$, the $i$-th element in the $j$-th color class of $A$.
The goal of the remaining gadgets is to ensure that when the lines $\text{HL}(s)$ and $\text{VL}(s)$ (with $s=(j-1)t+i$) are chosen, additional lines corresponding to selecting element $b_i^j$ of $B$ have to be expressly selected.
We define $\text{HL} := \{\text{HL}(s) $ $|$ $ s \in [tk]\}$ and $\text{VL} := \{\text{VL}(s) $ $|$ $ s \in [tk]\}$.

\begin{figure}
  \centering
  \begin{minipage}{0.45\textwidth}
    \centering
    \begin{tikzpicture}
    \def\s{3}

\foreach \i in {1,...,8}{
  \node[bp] at (\i / \s,\i / \s) {} ;
}
\foreach \i in {1,...,9}{
  \node at (-1,\i / \s - 0.5 / \s) {$a_\i$} ;
  \draw[densely dotted, thin] (-0.5,\i / \s - 0.5 / \s) -- (10 / \s,\i / \s - 0.5 / \s) ;
}

\node[rp] at (0.25 / \s,8 / \s + 0.75 / \s) {} ;
\node[rp] at (8 / \s + 0.75 / \s,0.25 / \s) {} ;

\draw[very thick] (-0.5,4 / \s - 0.5 / \s) -- (10 / \s,4 / \s - 0.5 / \s) ;
\draw[very thick] (4 / \s - 0.5 / \s,- 0.5 / \s) -- (4 / \s - 0.5 / \s,9.5 / \s) ;

  \end{tikzpicture}
    \subcaption{The interval gadget corresponding to $[a_1,a_9]=\{a_1,\ldots,a_9\}$. In thin dotted, the mapping between elements and potential lines. In bold, the choice of the lines corresponding to picking $a_4$. If one wants to separate these points with two lines, one almost horizontal and one almost vertical, the choice of the former imposes the latter.}
    \label{fig:interval-gadget}
  \end{minipage}
  \qquad
\begin{minipage}{0.45\textwidth}
  \centering
  \sv{
  \resizebox{135pt}{!}}{
  \begin{tikzpicture}
    \def\s{3}
    \def\d{0.5}
    \def\x{15}

\foreach \i in {1,...,\x}{
  \node[bp] at (\i / \s,\i / \s) {} ;
}

\foreach \i/\j in {0/4,1/6,2/7,3/9,5/10,6/12,7/13,9/14,10/15}{
  \node[rp] at (\i / \s + \d / \s - 0.25 / \s,\j / \s + \d / \s + 0.25 / \s) {} ;
  \node[rp] at (\j / \s + \d / \s + 0.25 / \s,\i / \s + \d / \s - 0.25 / \s) {} ;
}

\foreach \i in {4,7,14}{
  \draw[very thick] (\i / \s - 0.5 / \s,- 0.5 / \s) -- (\i / \s - 0.5 / \s,\x / \s + 1.5 / \s) ;
  \draw[very thick] (- 0.5 / \s,\i / \s - 0.5 / \s) -- (\x / \s + 1.5 / \s,\i / \s - 0.5 / \s) ;
}

  \end{tikzpicture}
  }
  \subcaption{The interval gadgets put together. A representation of one track. Separating these points with the fewest axis-parallel lines requires taking the horizontal and vertical lines associated to a minimum hitting set.}
  \label{fig:interval-gadgets}
\end{minipage}
\caption{To the left, one interval. To the right, several put together to form one track.}
\label{fig:interval}
\end{figure}
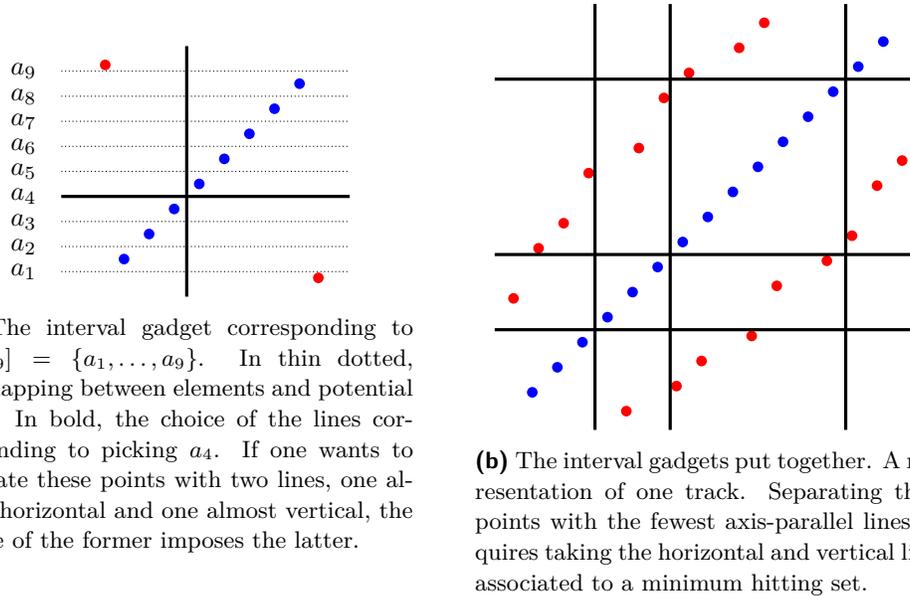

\medskip
\noindent
\textbf{Encoding inter-class permutation $\sigma$.}
To encode the permutation $\sigma$ of the $k$ color classes of $\mathcal I$, we allocate a square subgrid of the same dimension as the space used for the encoding of track $A$, roughly $4tk$-by-$4tk$, and we place it to the right of $A$ right as depicted in Figure~\ref{fig:overall}.
This square subgrid is naturally and regularly split into $k^2$ smaller square subgrids of equal dimension (roughly $4t$-by-$4t$).
This decomposition can be seen as the $k$ color classes of $\mathcal I$, or equivalently, the $k$-by-$k$ crossing\footnote{we use this term informally to avoid confusion with what we have been calling \emph{grids} so far.} obtained by drawing horizontal lines between two contiguous horizontal long alleys and vertical lines between two contiguous vertical long alleys.
We only put points in exactly one smaller square subgrid per column and per row.
Let $\sigma := \sigma(1)\sigma(2) \ldots \sigma(k)$ and $\text{Cell}(a,b)$ be the smaller square subgrid in the $a$-th row and $b$-th column of the $k$-by-$k$ crossing.
For each $j \in [k]$, we put in $\text{Cell}(j,\sigma(j))$ a diagonal of $t-1$ blue points and two red points corresponding to the full interval $[a^j_1,a^j_t]$ (see Figure~\ref{fig:permutation-classes}).
We denote by $\mathcal G(\sigma)$ those sets of red and blue points in the encoding of $\sigma$.
We surround $\mathcal G(\sigma)$ by $2k$ vertical long alleys similar to the $2k$ long alleys surrounding $\mathcal G(A)$.
Notice that $\mathcal G(\sigma)$ and $\mathcal G(A)$ share the same $2k$ surrounding horizontal long alleys.

The way the gadget $\mathcal G(\sigma)$ works is quite intuitive.
Given $k$ choices of horizontal lines originating from a separation in $\mathcal G(A)$ and a budget of $k$ extra lines for the separation within $\mathcal G(\sigma)$, the only option is to copy with the vertical line the choice of the horizontal line.
It results in a vertical propagation of the initial choices accompanied by the desired reordering of the \emph{color classes}.
The vertical line matching the choice of $\text{HL}(s)$ in the corresponding cell of $\mathcal G(\sigma)$ is denoted by $\text{VL}'(s)$. Let $\text{VL}' := \{\text{VL}'(s) $ $|$ $ s \in [tk]\}$.
Note that corresponding lines in $\text{VL}$ and in $\text{VL}'$ have a different order from left to right.

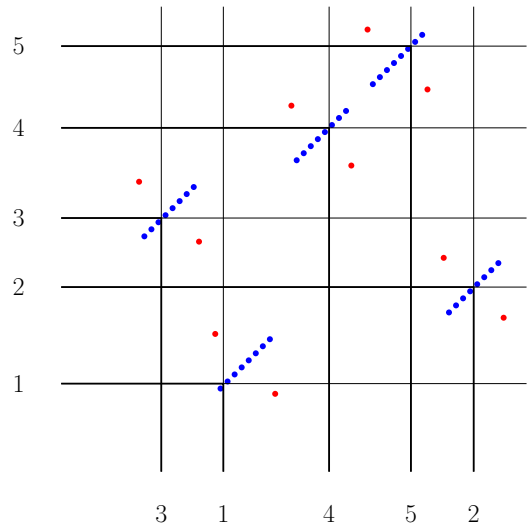
\begin{figure}[t!]
  \centering
  \resizebox{200pt}{!}{
  \begin{tikzpicture}
    \def\s{3}
    \def\t{8}
    \def\h{1.8}

    \foreach \j/\k in {1/3,2/1,3/4,4/5,5/2}{
      \begin{scope}[xshift=\h * \j cm,yshift=\h * \k cm,scale=0.5]
        \foreach \i in {1,...,\t}{
          \node[bp] at (\i / \s,\i / \s) {} ;
        }
        \node[rp] at (0.25 / \s,\t / \s + 0.75 / \s) {} ;
        \node[rp] at (\t / \s + 0.75 / \s,0.25 / \s) {} ;
      \end{scope}
    }

    \pgfmathtruncatemacro\xl{0}
    \pgfmathtruncatemacro\xr{35 / \s}

    \pgfmathtruncatemacro\yt{35 / \s}
    \pgfmathtruncatemacro\yb{0}

    \coordinate (a1) at (\xl,6.25 / \s) ;
    \coordinate (a2) at (\xl,13.1 / \s) ;
    \coordinate (a3) at (\xl,18 / \s) ;
    \coordinate (a4) at (\xl,24.4 / \s) ;
    \coordinate (a5) at (\xl,30.2 / \s) ;
    
    \draw (a1) -- (\xr,6.25 / \s) ;
    \draw (a2) -- (\xr,13.1 / \s) ;
    \draw (a3) -- (\xr,18 / \s) ;
    \draw (a4) -- (\xr,24.4 / \s) ;
    \draw (a5) -- (\xr,30.2 / \s) ;

    \coordinate (b1) at (7.1 / \s,\yb) ;
    \coordinate (b2) at (11.5 / \s,\yb) ;
    \coordinate (b3) at (19 / \s,\yb) ;
    \coordinate (b4) at (24.8 / \s,\yb) ;
    \coordinate (b5) at (29.25 / \s,\yb) ;

    \draw[very thin] (7.1 / \s,\yt) -- (b1) ;
    \draw[very thin] (11.5 / \s,\yt) -- (b2) ;
    \draw[very thin] (19 / \s,\yt) -- (b3) ;
    \draw[very thin] (24.8 / \s,\yt) -- (b4) ;
    \draw[very thin] (29.25 / \s,\yt) -- (b5) ;

    \draw[very thick]  (a1) -- (11.5 / \s,6.25 / \s) -- (b2) ;
    \draw[very thick]  (a2) -- (29.25 / \s,13.1 / \s) -- (b5) ;
    \draw[very thick]  (a3) -- (7.1 / \s,18 / \s) -- (b1) ;
    \draw[very thick]  (a4) -- (19 / \s,24.4 / \s) -- (b3) ;
    \draw[very thick]  (a5) -- (24.8 / \s,30.2 / \s) -- (b4) ;

    \foreach \i in {1,...,5}{
      \node[left of=a\i] {\LARGE{$\i$}} ;
    }
    \node[below of=b1] {\LARGE{$3$}} ;
    \node[below of=b2] {\LARGE{$1$}} ;
    \node[below of=b3] {\LARGE{$4$}} ;
    \node[below of=b4] {\LARGE{$5$}} ;
    \node[below of=b5] {\LARGE{$2$}} ;
  \end{tikzpicture}
  }
\caption{Encoding permutation $\sigma=31452$.
  The choices within the five color classes are transferred from almost horizontal lines to almost vertical ones.
  This way, we obtain the desired reordering of the color classes.}
\label{fig:permutation-classes}
\end{figure}

\medskip
\noindent
\textbf{Encoding of the intra-class permutations $\sigma_j$'s and track B.}
If the encoding of permutation $\sigma$ is conceptually simple, the number of intended lines separating red and blue points in $\mathcal G(\sigma)$ has to be linear in the number of permuted elements.
Since we wish to encode a permutation $\sigma_j$ (for every $j \in [k]$) on $t$ elements, we cannot use the same mechanism as it would blow-up our parameter dramatically and would not result in an FPT reduction.

For the gadget $\mathcal G_{\approx v}(\sigma_j)$ partially encoding the permutation $\sigma_j$, we will crucially use the fact that separating lines can have arbitrary slopes.
Slightly to the right (at distance at least $\ell$) of the vertical line bounding the right end of $\mathcal G(\sigma)$ and far in the south direction, we place a gadget $\mathcal G(B)$ encoding track B similarly to the encoding of track A up to some symmetry that we will make precise later.
We also incline the whole encoding of track B with a small, say 5, degree angle, in a way that its top-left corner is to the right of its bottom-left corner. We round up the real coordinates that this rotation incurs to rationals at distance less than, say, $(kt)^{-10}$.  
We denote by $\hat v$ the distance along the $y$-axis between $\mathcal G(\sigma)$ and $\mathcal G(B)$.
Eventually $\hat v$ will be chosen much larger than $\Theta(kt)$, which is the size of $\mathcal G(A)$, $\mathcal G(B)$, $\mathcal G(\sigma)$.
Below $\mathcal G(\sigma)$ at a distance $2 \hat v$ along the $y$-axis, we place gadgets $\mathcal G_{\approx v}(\sigma_j)$'s; from left to right, we place $\mathcal G_{\approx v}(\sigma_{\sigma(1)})$, $\mathcal G_{\approx v}(\sigma_{\sigma(2)})$, \dots, $\mathcal G_{\approx v}(\sigma_{\sigma(k)})$ such that for every $i \in [k]$, $G_{\approx v}(\sigma_{\sigma(i)})$ falls below the $i$-th column of the $k$-by-$k$ crossing of $\mathcal G(\sigma)$.
Gadgets $\mathcal G_{\approx v}(\sigma_j)$'s are represented by small round shapes in Figure~\ref{fig:overall}.
Notwithstanding what is drawn on the overall picture, the $\mathcal G_{\approx v}(\sigma_j)$'s can be all placed at the same $y$-coordinates.
Let $y_1 := y_0 - 2 \hat v$ (the exact value of $y_1$ is not crucial).
Also, we represent track B slanted by a 45 degree angle, instead of the actual 5 degree angle, to be able to fit everything on one page and convey the main ideas of the construction.
In general, for the figure to be readable, the true proportions are not respected.
The size of every gadget is much smaller than the distance between two different groups of gadgets, so that every line \emph{entering} a gadget traverses it in an axis-parallel fashion.   

Gadget $\mathcal G_{\approx v}(\sigma_j)$ is built in the following way.
For each $i \in [t]$ and $j \in [k]$, we draw a fictitious points $p_i^j$ corresponding to the intersection of a close to vertical line corresponding to picking element $b_i^j$ in gadget $\mathcal G(B)$ with the bottom end of $\mathcal G(B)$.
Read from left to right, the $p_i^j$'s have the same order as the $b_i^j$'s in $(B,\leq_B)$.
For every $s=(j-1)t+i$ (with $j \in [k]$ and $i \in [t]$), let $q_i^j$ be the point of $y$-coordinate $y_1$ on the line $\text{VL}'(s)$.
We define the line $\text{SL}(s)$ as going through $p_i^j$ and $q_i^j$, and set $\text{SL} := \{\text{SL}(s) $ $|$ $ s \in [tk]\}$.
We add two blue points just to the left and just to the right of $q_i^j$ at distance $\epsilon := z^{-10}$.
We also add two blue points on line $\text{SL}(s)$, one to the left of $q_i^j$ and one to the right of $q_i^j$.
Finally, we place two red points for each $\mathcal G_{\approx v}(\sigma_j)$ at the bottom-left and top-right of the gadget (see Figure~\ref{fig:permutation-within-a-class}).
Note that in the figure, the lines in $\text{SL}$ form a large angle with the $y$-axis, while in fact they are quite close to a 5 degree angle and behave like \emph{relatively vertical\footnote{By that, we mean that the lines are close to vertical for axes aligned with the encoding of track B.}} lines within $\mathcal G(B)$ (since $\mathcal G(B)$ is also inclined by 5 degrees).

Assuming that line $\text{VL}'(s=(j-1)t+i)$ has been selected, it might be observed from Figure~\ref{fig:permutation-within-a-class} that separating the red points from the blue points in $\mathcal G_{\approx v}(\sigma_j)$ with a budget of one additional line requires to take a line crossing $\text{VL}'(s)$ at (or very close to) $q_i^j$ and with a higher or equal slope to $\text{SL}(s)$.
It is not quite what we wanted.
What we achieved so far is only to link the \emph{choice of $a_i^j$} with the \emph{choice of an element smaller or equal to $b_i^j$}.
We will use a symmetry $\mathcal G_{\approx h}(\sigma_j)$ of gadget $\mathcal G_{\approx v}(\sigma_j)$ to get the other inequality so that choosing some lines corresponding to $a_i^j$ actually forces to take some lines corresponding to $b_i^j$.

\begin{figure}[h!]
  \centering
  \resizebox{400pt}{!}{
  \begin{tikzpicture}[
      potential/.style={densely dotted, thin},
      extended line/.style={shorten <=-#1},
      extended line/.default=6.5cm]

    \node (offset) at (-5,0) {} ;
    
    \def\s{0.5}
    \def\t{8}
    \def\ytop{10}
    \def\ybot{-3}
    \def\ygad{3}
    \def\offset{15}
    \foreach \i in {1,...,\t}{
      \node at (\s * \i, \ytop+0.5) {\Large{$\i$}} ;
      \draw[potential] (\s * \i,\ytop) -- (\s * \i,\ybot) ;
      \node[bp] at (\s * \i-0.1,\ygad) {} ;
      \node[bp] at (\s * \i+0.1,\ygad) {} ;
    }
    
    \foreach \i/\j in {0/0.22,1/0.26,2/0.3,3/0.23,4/0.29,5/0.26,6/0.51,7/0.44}{
      \node[bp] at (\s * \i+0.1,\ygad-\j) {} ;
    }
    \foreach \i/\j in {2/0.22,3/0.24,4/0.29,5/0.23,6/0.29,7/0.28,8/0.47,9/0.37}{
      \node[bp] at (\s * \i-0.1,\ygad+\j) {} ;
    }
    
    \foreach \i/\j in {1/6,2/3,3/2,4/8,5/5,6/7,7/1,8/4}{
      \coordinate (a\i) at (\s * \i,\ygad) ;
      \coordinate (b\j) at (\s * \j+\s *\offset,\ytop) ;
      \node at (\s * \j+\s *\offset,\ytop+0.5) {\Large{$\i$}} ;
      \draw[potential,extended line]  (a\i) -- (b\j) ;
    }

    \node[rp] at (0.3,-2.25) {} ;
    \node[rp] at (4.2,6) {} ;

    \draw[very thick] (\s * 6,\ytop) -- (\s * 6,\ybot) ;
    \draw[very thick,extended line]  (a6) -- (b7) ;
    \fill[opacity=0.2] (a6) -- (b7) -- (b1) -- cycle;

    \draw[thick,->] (-1.1,-0.8) -- (-0.7,-1.2) ;
    \draw[thick,->] (8.5,8) -- (8.1,8.4) ;
  \end{tikzpicture}
  }
  \caption{Half-encoding of permutation $\sigma_j=73285164$ of the $j$-th color class.
    Observe that the choice of the, say, sixth almost horizontal candidate line only forces to take the slanted line depicted in bold or a line having the same intersection with the almost horizontal line but a larger slope.
  For the sake of legibility, the angles between the vertical lines and the slanted lines are exaggerated.}
  \label{fig:permutation-within-a-class}
\end{figure}
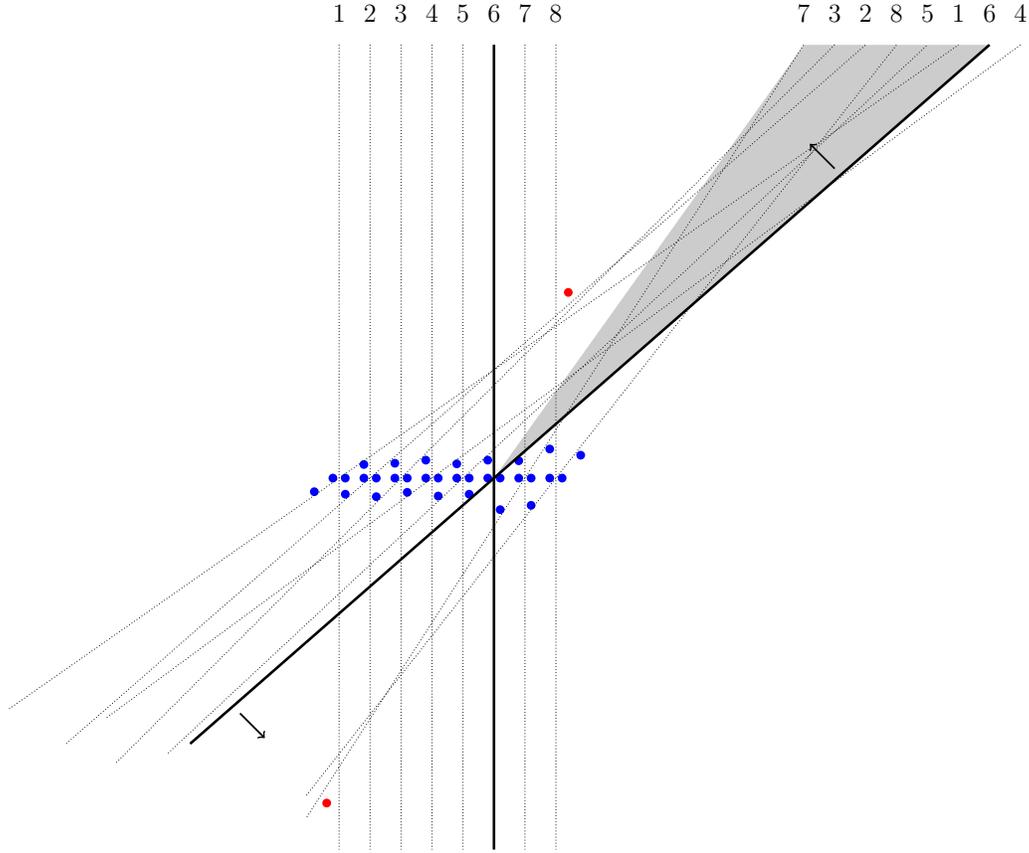

We add a gadget $\mathcal G(\text{id})$ below the $\mathcal G_{\approx v}(\sigma_j)$'s.
$\mathcal G(\text{id})$ is obtained by mimicking $\mathcal G(\sigma)$ for the identity permutation.
We surround $\mathcal G(\text{id})$ by $2k$ new horizontal long alleys.
The horizontal line matching the choice of $\text{VL}'(s)$ in $\mathcal G(\text{id})$ is denoted by $\text{HL}'(s)$.
At a distance $\hat h \approx \hat v/(cos(5^\circ) \cdot sin(5^\circ))$ to the right of $\mathcal G(\text{id})$ we place gadgets $\mathcal G_{\approx h}(\sigma_j)$'s analogously to the $\mathcal G_{\approx v}(\sigma_j)$'s.
The fictitious points ${p'}_i^j$ (analogous of $p_i^j$) used for the construction of the lines $\text{SL}'(s)$ (analogous of $\text{SL}(s)$) are located at the right end of $\mathcal G(B)$ and ordered as $B$ when read from top to bottom.
The slight difference in the construction of $\mathcal G(B)$ from the $B$-intervals (compared to $\mathcal G(A)$ from the $A$-intervals) is that the diagonal of blue points go from the top-left corner to the bottom-right corner (instead of bottom-left to top-right).
Similarly to our previous definitions, we define $\text{HL}' := \{\text{HL}'(s) $ $|$ $ s \in [tk]\}$ and $\text{SL}' := \{\text{SL}'(s) $ $|$ $ s \in [tk]\}$.
Note that the choice of $\hat h$ makes the lines of $\text{SL}'$ form a close to $5$ degree angle with the $x$-axis and so \emph{arrive relatively horizontal} within $\mathcal G(B)$.

\medskip
\noindent
\textbf{Putting the pieces together.}
We already hinted at how the different gadgets are combined together.
We choose the different typical values so that: $kt \ll \hat v < \hat h \ll z$.
For instance, $\hat v := 100((kt)^2+1)$ and $z := 100(\hat h^5+1)$.
An important and somewhat hidden consequence of $z$ being much greater than $\hat v$ and $\hat h$ is that the bulk of the construction (say, all the gadgets which are not long alleys) occupies a tiny space in the top-left corner of $\Gamma$.
We set the length $\ell$ of the long alleys to $100(k^2+1)$.
Point $(x_0^A,y_0^A)$ corresponds to the bottom-left corner of the square in bold with a diagonal close to the overall top-left corner.

Slightly outside grid $\Gamma$ we place $14$ pairs of long alleys ($7$ horizontal and $7$ vertical) of width, say, $(kt)^{-10}$ to force the $14$ lines in bold in Figure~\ref{fig:overall}.
Note that, on the figure, we do not explicitly represent those long alleys but only the lines they force.
The purpose of those new long alleys is to separate groups of gadgets from each other.
Going clockwise all around the grid $\Gamma$, we alternate red-blue and blue-red alleys so that two consecutive long alleys do not need a further separation.
The even parity of those alleys make this alternation possible. 
Each one of the $64$ faces that those $14$ lines define is called a \emph{super-cell}.

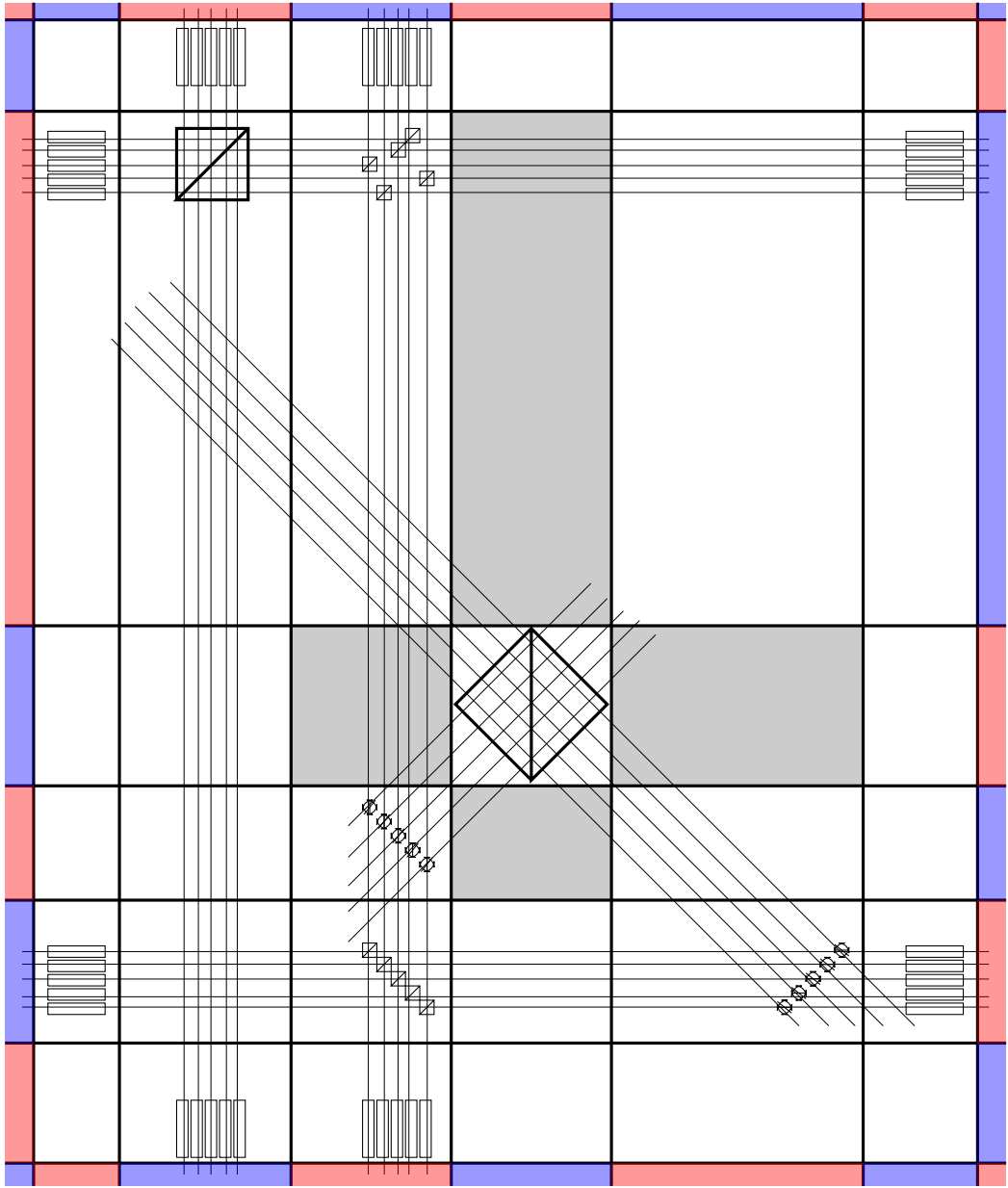
\begin{figure}[h!]
  \centering
  \sv{
  \resizebox{270pt}{!}}{
  \begin{tikzpicture}[scale=0.79]
    \def\d{0.25}
    \def\k{5}
    \def\e{1.5}
    \def\a{1}
    \def\ab{0.8}
    \def\f{0.4}

\begin{scope}[yshift=8cm]
    
    \foreach \i/\j in {-21/1,-21/2,-21/3,-21/4,-21/5,39/1,39/2,39/3,39/4,39/5}{
      \begin{scope}[xshift=\d * \i cm,yshift=\d * \j cm]
        \draw (0,0) -- (0,\ab * \d) -- (\a,\ab * \d) -- (\a,0) -- (0,0) ;
      \end{scope}
    }
    \foreach \i/\j in {-12/9,-11/9,-10/9,-9/9,-8/9,-12/-66,-11/-66,-10/-66,-9/-66,-8/-66,1/9,2/9,3/9,4/9,5/9,1/-66,2/-66,3/-66,4/-66,5/-66}{
      \begin{scope}[xshift=\d * \i cm,yshift=\d * \j cm]
        \draw (0,0) -- (\ab * \d,0) -- (\ab * \d,\a) -- (0,\a) -- (0,0) ;
      \end{scope}
    }

    \foreach \i/\j in {1/3,2/1,3/4,4/5,5/2}{
      \begin{scope}[xshift=\d * \i cm,yshift=\d * \j cm]
        \draw (0,0) -- (0,\d) -- (\d,\d) -- (\d,0) -- (0,0) -- (\d,\d) ;
      \end{scope}
    }

        \begin{scope}[very thick, xshift=-3 cm,yshift=\d cm]
      \draw (0,0) -- (0,\k * \d) -- (\k * \d,\k * \d) -- (\k * \d,0) -- (0,0) -- (\k * \d,\k * \d) ;
    \end{scope}

\end{scope}

    \foreach \i/\j in {-21/-20,-21/-21,-21/-22,-21/-23,-21/-24,39/-20,39/-21,39/-22,39/-23,39/-24}{
      \begin{scope}[xshift=\d * \i cm,yshift=\d * \j cm]
        \draw (0,0) -- (0,\ab * \d) -- (\a,\ab * \d) -- (\a,0) -- (0,0) ;
      \end{scope}
    }

 \foreach \i/\j in {1/-20,2/-21,3/-22,4/-23,5/-24}{
      \begin{scope}[xshift=\d * \i cm,yshift=\d * \j cm]
        \draw (0,0) -- (0,\d) -- (\d,\d) -- (\d,0) -- (0,0) -- (\d,\d) ;
      \end{scope}
    }

    \foreach \i/\j in {1/-10,2/-11,3/-12,4/-13,5/-14}{
      \begin{scope}[xshift=\d * \i cm,yshift=\d * \j cm]
        \draw[rounded corners] (0,0) -- (0,\d) -- (\d,\d) -- (\d,0) -- cycle ;
        \draw (\d/2,0) -- (\d/2,\d) ;
        \draw (\d/2 - \f * \d,0) -- (\d/2+ \f * \d,\d) ;
      \end{scope}
    }

    \foreach \i/\j in {34/-20,33/-21,32/-22,31/-23,30/-24}{
      \begin{scope}[xshift=\d * \i cm,yshift=\d * \j cm]
        \draw[rounded corners] (0,0) -- (0,\d) -- (\d,\d) -- (\d,0) -- cycle ;
        \draw (0,\d/2) -- (\d,\d/2) ;
        \draw (0,\d/2 + \f * \d) -- (\d,\d/2 - \f * \d) ;
      \end{scope}
    }

     \begin{scope}[very thick,xshift=3.2 cm,yshift=-1.9 cm,rotate=45]
      \draw (0,0) -- (0,\e * \k * \d) -- (\e * \k * \d,\e * \k * \d) -- (\e * \k * \d,0) -- (0,0) -- (\e * \k * \d,\e * \k * \d) ;
     \end{scope}

     \def\el{-5.7}
     \def\er{11.2}
     \def\top{11.6}
     \def\bot{-8.8}
 
     \foreach \i in {9.31,9.12,8.85,8.63,8.38,-4.9,-5.12,-5.38,-5.69,-5.87}{
       \draw[very thin] (\el,\i) -- (\er,\i) ;
     }
     \foreach \i in {-2.87,-2.62,-2.4,-2.13,-1.94,0.35,0.63,0.87,1.06,1.38}{
       \draw[very thin] (\i,\top) -- (\i,\bot) ;
     }

     \draw[very thin] (0,-2.7) -- +(45:6) ;
     \draw[very thin] (0,-3.25) -- +(45:6.4) ;
     \draw[very thin] (0,-3.75) -- +(45:6.8) ;
     \draw[very thin] (0,-4.2) -- +(45:7.2) ;
     \draw[very thin] (0,-4.73) -- +(45:7.6) ;

     \draw[very thin] (7.88,-6.2) -- +(135:17) ;
     \draw[very thin] (8.4,-6.2) -- +(135:17.4) ;
     \draw[very thin] (8.86,-6.2) -- +(135:17.8) ;
     \draw[very thin] (9.35,-6.2) -- +(135:18.15) ;
     \draw[very thin] (9.9,-6.2) -- +(135:18.4) ;

     \def\sxb{-6}
     \def\sxe{11.5}
     \def\syb{11.7}
     \def\sye{-9}

     \foreach \j in {-2,-4,-6.5,-8.6,0.8,9.8,11.4}{
       \draw[very thick] (\sxb,\j) -- (\sxe,\j) ;
     }
     \foreach \i in {-5.5,-4,-1,1.8,4.6,9,11}{
       \draw[very thick] (\i,\syb) -- (\i,\sye) ;
     }

     \def\opa{0.2}
     \def\opacol{0.4}

     \fill[opacity=\opa] (-1,0.8) -- (1.8,0.8) -- (1.8,-2) -- (-1,-2) -- cycle ;
     \fill[opacity=\opa] (4.6,0.8) -- (9,0.8) -- (9,-2) -- (4.6,-2) -- cycle ;
     \fill[opacity=\opa] (4.6,-2) -- (4.6,-4) -- (1.8,-4) -- (1.8,-2) -- cycle ;
     \fill[opacity=\opa] (4.6,0.8) -- (4.6,9.8) -- (1.8,9.8) -- (1.8,0.8) -- cycle ;

     \foreach \i/\j in {-5.5/-4,-1/1.8,4.6/9,11/\sxe}{
       \fill[blue,opacity=\opacol] (\j,\syb) -- (\i,\syb) -- (\i,11.4) -- (\j,11.4) -- cycle ;
       \fill[red,opacity=\opacol] (\j,\sye) -- (\i,\sye) -- (\i,-8.6) -- (\j,-8.6) -- cycle ;
     }
     \foreach \i/\j in {\sxb/-5.5,-4/-1,1.8/4.6,9/11}{
       \fill[red,opacity=\opacol] (\j,\syb) -- (\i,\syb) -- (\i,11.4) -- (\j,11.4) -- cycle ;
       \fill[blue,opacity=\opacol] (\j,\sye) -- (\i,\sye) -- (\i,-8.6) -- (\j,-8.6) -- cycle ;
     }
     \foreach \i/\j in {-8.6/-6.5,-4/-2,0.8/9.8}{
       \fill[red,opacity=\opacol] (\sxb,\i) -- (\sxb,\j) -- (-5.5,\j) -- (-5.5,\i) -- cycle ;
       \fill[blue,opacity=\opacol] (\sxe,\i) -- (\sxe,\j) -- (11,\j) -- (11,\i) -- cycle ;
     }
     \foreach \i/\j in {-6.5/-4,-2/0.8,9.8/11.4}{
       \fill[blue,opacity=\opacol] (\sxb,\i) -- (\sxb,\j) -- (-5.5,\j) -- (-5.5,\i) -- cycle ;
       \fill[red,opacity=\opacol] (\sxe,\i) -- (\sxe,\j) -- (11,\j) -- (11,\i) -- cycle ;
     }
  \end{tikzpicture}
  }
  \caption{The overall picture.
    The thin rectangles are long alleys, the bold large squares with a diagonal are the encoding of track A, in the top left corner, and track B, slanted by $45$ degrees (for the sake of fitting the whole construction on one page; in reality the encoding of B is only inclined by $5$ degrees). The smaller squares with a diagonal are simple interval gadgets and the small round gadgets are half-encodings of the permutations $\sigma_i$'s.
    The four super-cells filled with grey contain $4k$ long alleys slanted by $5$ degrees.
  The (super-)cells filled with red and blue match their color, and are monochromatic once the $14$ lines imposed by the outermost long alleys have been selected. }
  \label{fig:overall}
\end{figure}

The four lines in bold surrounding $\mathcal G(B)$ are close (say, at distance $10t$) to the north, south, west, and east ends of that gadget.
On the four super-cells adjacent to the super-cell containing $\mathcal G(B)$, shown in gray, we place $4k$ long alleys each of width $4t-4$, analogously to what was done for $\mathcal G(A)$, but slanted by a $5$ degree angle (as the gadget $\mathcal G(B)$).
As for track A, these alleys force, relatively to the orientation of $\mathcal G(B)$, one \emph{close to horizontal} line and one \emph{close to vertical line} per color class.
The long alleys are placed just next to $\mathcal G(B)$ and are not crossed by any other candidate lines.

This finishes the construction. We ask for a
separation of $\mathcal R$ and $\mathcal B$ with $6k+14$ lines.  \sv{The
correctness of the reduction is deferred to the appendix.} \lv{
We now show the correctness of the reduction.

\medskip

\textbf{If $\mathcal I$ is a YES-instance for \shss, then $6k+14$ lines are sufficient.}
Let $F$ be the set of $14$ lines forced by the outermost long alleys (lines in bold in Figure~\ref{fig:overall}).
Let $(a_{u_1}^1,b_{u_1}^1), (a_{u_2}^2,b_{u_2}^2), \ldots, (a_{u_k}^k,b_{u_k}^k)$ be a solution of \shss ($u_1, u_2, \ldots, u_k \in [t]$).
Let $s_j := (j-1)t+u_j$ for every $j \in [k]$.
$F \cup \bigcup_{j \in [k]}\{\text{HL}(s_j), \text{VL}(s_j), \text{VL}'(s_j), \text{HL}'(s_j), \text{SL}(s_j), \text{SL}'(s_j)\}$ is a set of $6k+14$ lines.
We claim that it is a solution.

Due to $F$, we only need to check that the red and blue points of the \emph{same} super-cell are separated.
The constant number of outermost long alleys are well separated: see the alternating coloring of Figure~\ref{fig:overall}.
As the other long alleys also alternates red-blue and blue-red, the super-cells containing $k$ long alleys are all well separated.

This leaves us $6$ super-cells to check: namely those of $\mathcal G(A)$, $\mathcal G(B)$, $\mathcal G(\sigma)$, $\mathcal G(\text{id})$, the $\mathcal G_{\approx v}(\sigma_j)$'s, and the $\mathcal G_{\approx h}(\sigma_j)$'s.
The points in $\mathcal G(\sigma)$ and $\mathcal G(\text{id})$ are separated as in Figure \ref{fig:interval-gadget}, since the choice of $\text{VL}'(s_j)$ matches the choices of $\text{HL}(s_j)$ and $\text{HL}'(s_j)$.
As it can be observed by looking at $\text{Cell}(4,3)$ and $\text{Cell}(5,4)$ of Figure~\ref{fig:permutation-classes}, there is not interaction between the red and blue points of diagonally adjacent faces of the $k$-by-$k$ crossing (in $\mathcal G(\sigma)$ and $\mathcal G(\text{id})$).

Since $a_{u_1}^1, a_{u_2}^2, \ldots, a_{u_k}^k$ (resp.~$b_{u_1}^1, b_{u_2}^2, \ldots, b_{u_k}^k$) is a hitting set of $\mathcal S_A$ (resp.~$\mathcal S_B$), the points in $\mathcal G(A)$ (resp.~$\mathcal G(B)$) are separated as in Figure \ref{fig:interval-gadgets}.
Indeed for each interval $I \in \mathcal S_A$ (resp~$I \in \mathcal S_B$), there is an $j \in [k]$ such that $a_{u_j}^j$ hits $I$ (resp.~$b_{u_j}^j$ hits $I$), and the two red points encoding $I$ are in the two quadrants defined by $\text{HL}(s_j)$ and $\text{VL}(s_j)$ (resp.~defined by $\text{SL}(s_j)$ and $\text{SL}'(s_j)$)  where there is no blue point.

Similarly the two red points of a gadget $\mathcal G_{\approx v}(\sigma_j)$ (resp.~$\mathcal G_{\approx h}(\sigma_j)$) are separated from the blue points: they are in the two regions defined by $\text{VL}'(s_j)$ and $\text{SL}(s_j)$ (resp.~$\text{HL}'(s_j)$ and $\text{SL}'(s_j)$) where there is no blue point.
Two consecutive gadgets $\mathcal G_{\approx v}(\sigma_j)$ and $\mathcal G_{\approx v}(\sigma_{j+1})$ (resp.~$\mathcal G_{\approx h}(\sigma_j)$ and $\mathcal G_{\approx h}(\sigma_{j+1})$) do not interact.
In fact, all the blue points land in the quadrangular faces touching two consecutive gadgets.

\medskip

\textbf{If $6k+14$ lines are sufficient, then $\mathcal I$ is a YES-instance for \shss.}
Let $S$ be a feasible solution consisting of at most $6k+14$ lines.
The lines of $S$ should separate \emph{all} the straight-line segments whose one extremity is at a red point and the other is at a blue point.
We call such a segment a \emph{red/blue segment} or a \emph{red/blue pair} (or simply \emph{pair}).

First, we can assume that $F \subseteq S$, where $F$ is the set of $14$ lines forced by the $28$ outermost long alleys.
Indeed, in each of those long alleys there should be a line of $S$ separating at least two red/blue segments, such that the two segments and the line have \emph{not} a common intersection.
For every line $\mathcal L$ satisfying this property, the line in $F$ responsible from separating this long alley separates a superset of the red/blue pairs separated by $\mathcal L$; and therefore can be chosen.

We will now focus on a particular subset of red/blue pairs.
Consider the set $\mathcal X$ of the red/blue segments within each of the $12k$ remaining long alleys between two points with the same $x$-coordinate (resp.~$y$-coordinate) in a horizontal alley (resp.~vertical alley), and by generalizing in the natural way this notion for the close to horizontal (resp.~vertical) alleys surrounding $\mathcal G(B)$.
There are $\ell$ such red/blue pairs per long alley, hence $|\mathcal X|=12k\ell$.
We partition the $12k$ long alleys into eight \emph{groups}: $\mathcal A_W$, $\mathcal A_E$, $\mathcal A_N$, $\mathcal A_S$, the axis-parallel long alleys to the west, east, north, and respectively, south of $\Gamma$, and $\mathcal B_W$, $\mathcal B_E$, $\mathcal B_N$, $\mathcal B_S$ the slightly slanted long alleys to the west, east, north, and respectively, south of $\mathcal G(B)$.

\begin{lemma}\label{lem:long-alleys}
  No line separates strictly more than $2\ell$ red/blue pairs of $\mathcal X$.
  Furthermore, the only way for a line to separate $2\ell$ red/blue pairs of $\mathcal X$ is to separate all the red/blue pairs of $\mathcal X$ of two long alleys belonging to a pair in $\{(\mathcal A_W,\mathcal A_E),(\mathcal A_N,\mathcal A_S),(\mathcal B_W,\mathcal B_E),(\mathcal B_N,\mathcal B_S)\}$ (and no other pair of $\mathcal X$).
\end{lemma}
\begin{proof}
Within the same group of long alleys, a line separates at most $\ell$ red/blue pairs of $\mathcal X$.
Indeed, say, the group of long alleys consists of horizontal alleys.
Then a line cannot separate two red/blue pairs sharing the same $x$-coordinate.
Furthermore, it can be observed that to separate within the same group exactly $\ell$ red/blue pairs of $\mathcal X$, the line has to separate the red/blue pairs of the \emph{same} long alley.

We also observe that a line intersects a positive number of red/blue pairs of $\mathcal X$ in at most two groups among $\mathcal A_W$, $\mathcal A_E$, $\mathcal A_N$, and $\mathcal A_S$ (resp.~$\mathcal B_W$, $\mathcal B_E$, $\mathcal B_N$, and $\mathcal B_S$) and at most three of the eight groups.

If a line intersects red/blue pairs of $\mathcal X$ in three groups, then those groups have to be (a) $\mathcal B_W$, $\mathcal B_E$, and $\mathcal A_W$, or (b) $\mathcal B_W$, $\mathcal B_E$, and $\mathcal A_E$, or (c) $\mathcal B_N$, $\mathcal B_S$, and $\mathcal A_N$, or (d) $\mathcal B_N$, $\mathcal B_S$, and $\mathcal A_S$.
Here we use the fact that $\hat h \ll z$.
Hence, all the other gadgets are much closer to the long alleys in $\mathcal A_W$ and $\mathcal A_N$ than to the long alleys in $\mathcal A_E$ and $\mathcal A_S$.
Thus, a line separating red/blue pairs in, say, $\mathcal A_E$ and $\mathcal B_E$ looks horizontal between $\mathcal B_E$  and the west end of $\Gamma$, and therefore cannot separate red/blue pairs in $\mathcal A_W$.

The cases (a), (b), (c), and (d) being symmetric, we only treat case (a).
A line corresponding to case (a), cannot separate $2\ell$ red/blue pairs of $\mathcal X$.
Here we use the fact that the distance between two groups of gadgets is much larger than the size of the gadgets.
So a line $\mathcal L$ separating some red/blue pairs in $\mathcal A_W$ and $\mathcal B_W$ looks horizontal between $\mathcal B_W$ and $\mathcal B_E$.
As the long alleys of $\mathcal B_W$ and $\mathcal B_E$ are slanted by a 5 degree angle, $\mathcal L$ cannot separate more than $100k<\ell$ red/blue pairs of $\mathcal X$ in $\mathcal B_W \cup \mathcal B_E$.
Indeed, a close to horizontal line cannot separate more than a constant (smaller than $50$) number of red/blue pairs of $\mathcal X$ per long alley of $\mathcal B_W \cup \mathcal B_E$.

At this point, one can eventually observe that the only ways to separate $2\ell$ red/blue pairs of $\mathcal X$ with one line, is to separate $\ell$ pairs in $\mathcal A_W$ (resp.~$\mathcal B_W$) and $\ell$ pairs in $\mathcal A_E$ (resp.~$\mathcal B_E$), or $\ell$ pairs in $\mathcal A_N$ (resp.~$\mathcal B_N$) and $\ell$ pairs in $\mathcal A_S$ (resp.~$\mathcal B_S$).
By a previous remark, the separated pairs within a group come from the \emph{same} long alley.
\end{proof}

As the remaining budget is $6k$ lines, it follows from Lemma~\ref{lem:long-alleys} that \emph{all} the lines of $S \setminus F$ have to separate exactly $2 \ell$ pairwise-disjoint red/blue pairs of $\mathcal X$.
Furthermore, in $S \setminus F$, there are $2k$ almost horizontal lines separating one long alley in $\mathcal A_W$ and the other in $\mathcal A_E$, $2k$ almost vertical lines separating one long alley in $\mathcal A_N$ and the other in $\mathcal A_S$, $k$ lines separating one long alley in $\mathcal B_W$ and the other in $\mathcal B_E$, and $k$ lines separating one long alley in $\mathcal B_N$ and the other in $\mathcal B_S$.

Let us draw a small parenthesis.
Despite what is represented in Figure~\ref{fig:overall}, the line of $S$ separating the $h$-th topmost long alley of $\mathcal A_W$ (resp.~the $v$-th leftmost long alley of $\mathcal A_N$) does not necessarily separate the $h$-th topmost long alley of $\mathcal A_E$ (resp.~the $v$-th leftmost long alley of $\mathcal A_S$).
Instead, this line separates \emph{one} long alley of $\mathcal A_E$ (resp. $\mathcal A_S$); it does not matter which one.
Therefore, the exact position of the long alleys of $\mathcal A_E \cup \mathcal A_S$ is not crucial.
What is important is that there are $2k$ horizontal long alleys very far east, and $2k$ vertical long alleys very far south. 
We nevertheless chose to align those alleys with the ones in $\mathcal A_W \cup \mathcal A_N$, since we think it leads to a more intuitive construction for the reader.
This closes the parenthesis.

Let us focus on the $k$ lines of $S$ separating the $k$ topmost long alleys of $\mathcal A_W$.
For each $j \in [k]$, we denote by $\mathcal L_j$ the one separating the $j$-th bottommost of those $k$ long alleys.
As we already observed those lines behave like horizontal lines in the smallest subgrid enclosing all the gadgets which are \emph{not} in $\mathcal A_E \cup \mathcal A_S$ (nor the $14$ outermost long alleys).
For each $j \in [k]$, let $a^j_{u_j}$ be the element of $A$ corresponding to $\mathcal L_j$ (with the correspondence described in Figure~\ref{fig:interval-gadget}).
In particular, by the position of the $k$ topmost long alleys of $\mathcal A_W$, it is indeed true that the $k$ lines $\mathcal L_1, \mathcal L_2, \ldots, \mathcal L_k$ translates to exactly one element per color class of track $A$.
We show, thanks to the following lemma, that $a^1_{u_1}$, $a^2_{u_2}, \ldots, a^k_{u_k}$ is a hitting set of $(A,\mathcal S_A)$.

\begin{lemma}\label{lem:simple-interval}
The only ways to separate a simple interval gadget with one horizontal line and one vertical line is to make them meet at the diagonal defined by the blue points.  
\end{lemma}
\begin{proof}
  If the lines meet \emph{above} the diagonal, then the bottom red point is not separated from the blue point just to the right of the vertical line.
  If the lines meet \emph{below} the diagonal, then the top red point is not separated from the blue point just to the left of the vertical line.
\end{proof}

Recall that we added for convenience the pairs of red points $R([a^j_1,a^j_t])$, for each $j \in [k]$.
We consider the simple interval gadget that each pair induces, that is, the two red points and the diagonal of blue points contained in the smallest square subgrid enclosing them.
Because of the long alleys in $\mathcal A_W$ and $\mathcal A_N$, we have a budget of exactly one horizontal line and one vertical line to separate each of those $k$ simple intervals.
By Lemma~\ref{lem:simple-interval}, the $k$ vertical lines of $S$ separating the $k$ leftmost long alleys of $\mathcal A_N$ have to \emph{agree} with the choices of the horizontal lines $\mathcal L_j$'s.
More formally, the $j$-th bottommost horizontal line intersects the $j$-th leftmost vertical line at the diagonal defined by the blue points.

This implies that all the intervals of $\mathcal S_A$ are hit by the $a^j_{u_j}$'s.
Indeed, if an interval $I$ is not hit, the smallest square subgrid $\gamma_I$ enclosing the corresponding pair of red points would not be intersected by $S$; and those red points would not be separated from any diagonal blue point in $\gamma_I$.

We now show that the choice of the lines corresponding to the elements $a^1_{u_1}$, $a^2_{u_2}, \ldots, a^k_{u_k}$ will force to take the lines corresponding to the elements $b^1_{u_1}$, $b^2_{u_2}, \ldots, b^k_{u_k}$.
Still by Lemma~\ref{lem:simple-interval}, $\mathcal G(\sigma)$ transmits the choices of the $\mathcal L_j$'s \emph{downwards} with the desired permutation $\sigma$ of the color classes, while $\mathcal G(\text{id})$ transmits unchanged the choices of the vertical lines separating $\mathcal G(\sigma)$ \emph{to the left}. 

Similarly to the simple argument of Lemma~\ref{lem:simple-interval}, once the axis-parallel line has been selected in a gadget $\mathcal G_{\approx v}(\sigma_j)$ or $\mathcal G_{\approx h}(\sigma_j)$, to separate the two red points from the four blue points on or close to the intended line (that is, $\text{SL}(s)$ when $\text{VL}'(s)$ has been selected, or $\text{SL}'(s)$ when $\text{HL}'(s)$ has been selected), one should choose the intended line itself or any line having the same intersection with the axis-parallel line and closer to this axis (see Figure~\ref{fig:permutation-within-a-class}).
The way the gadgets $\mathcal G_{\approx v}(\sigma_j)$'s, $\mathcal G_{\approx h}(\sigma_j)$'s, and $\mathcal G(B)$ are placed, it results in, for each color class $j$ of track B, selecting a relatively horizontal line somewhere to the left of the line corresponding to $b^j_{u_j}$, and selecting a relatively vertical line somewhere below the line corresponding to $b^j_{u_j}$ (see Figure~\ref{fig:2ineq}).

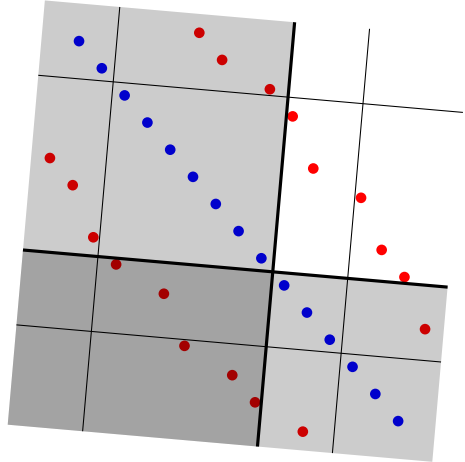
\begin{figure}[h!]
  \centering
  \begin{tikzpicture}[rotate=85]
  \def\s{3}
    \def\d{0.5}
    \def\x{15}

\foreach \i in {1,...,\x}{
  \node[bp] at (\i / \s,\i / \s) {} ;
}

\foreach \i/\j in {0/4,1/6,2/7,3/9,5/10,6/12,7/13,9/14,10/15}{
  \node[rp] at (\i / \s + \d / \s - 0.25 / \s,\j / \s + \d / \s + 0.25 / \s) {} ;
  \node[rp] at (\j / \s + \d / \s + 0.25 / \s,\i / \s + \d / \s - 0.25 / \s) {} ;
}

\foreach \i in {7}{
  \draw[very thick] (\i / \s - 0.5 / \s,- 0.5 / \s) -- (\i / \s - 0.5 / \s,\x / \s + 1.5 / \s) ;
  \draw[very thick] (- 0.5 / \s,\i / \s - 0.5 / \s) -- (\x / \s + 1.5 / \s,\i / \s - 0.5 / \s) ;

  \fill[opacity=0.2] (\i / \s - 0.5 / \s,- 0.5 / \s) -- (\i / \s - 0.5 / \s,\x / \s + 1.5 / \s)
                  -- (\i / \s - 7.5 / \s,\x / \s + 1.5 / \s) -- (\i / \s - 7.5 / \s,- 0.5 / \s) ;
  \fill[opacity=0.2] (- 0.5 / \s,\i / \s - 0.5 / \s) -- (\x / \s + 1.5 / \s,\i / \s - 0.5 / \s)
                  -- (\x / \s + 1.5 / \s,\i / \s + 9.5 / \s) -- (- 0.5 / \s,\i / \s + 9.5 / \s);
}
\foreach \i in {4,14}{
  \draw[thin] (\i / \s - 0.5 / \s,- 0.5 / \s) -- (\i / \s - 0.5 / \s,\x / \s + 1.5 / \s) ;
  \draw[thin] (- 0.5 / \s,\i / \s - 0.5 / \s) -- (\x / \s + 1.5 / \s,\i / \s - 0.5 / \s) ;
}

  \end{tikzpicture}
  \caption{In bold, the horizontal and vertical lines in $\mathcal G(B)$ corresponding to selecting some element $b_i^j$ of $B$.
    The grey regions materialize the potential positions for the slanted line in $\mathcal G_{\approx v}(\sigma_j)$ and the slanted line in $\mathcal G_{\approx h}(\sigma_j)$ once the lines corresponding to selecting $a_i^j$ have been chosen.}
  \label{fig:2ineq}
\end{figure}

Though, by Lemma~\ref{lem:simple-interval}, those two lines have to meet at the diagonal formed by the blue points.
The only way to realize that is that both lines agree on the choice of $b^j_{u_j}$.
This concludes to prove that choosing the lines corresponding to $a^1_{u_1}, a^2_{u_2}, \ldots, a^k_{u_k}$ to separate $\mathcal G(A)$ forces to select the lines corresponding to $b^1_{u_1}, b^2_{u_2}, \ldots, b^k_{u_k}$ to separate $\mathcal G(B)$.
Finally, as we already observed for track A, the $b^j_{u_j}$'s have to be a hitting set of $(B,\mathcal S_B)$; otherwise, the non hit interval would induce some non separated red/blue pairs.

As $a^1_{u_1}, a^2_{u_2}, \ldots, a^k_{u_k}$ is a hitting set of $(A,\mathcal S_A)$ and $b^1_{u_1}, b^2_{u_2}, \ldots, b^k_{u_k}$ is a hitting set of $(B,\mathcal S_B)$, $\mathcal I$ is a YES-instance.
 } \end{proof}

\section{FPT Algorithm Parameterized by Size of Smaller Set}
\label{sec:fpt}

We present a simple FPT algorithm for \hvrbs parameterized by  $\min\{|\mathcal R|,|\mathcal B|\}$.
In the following, w.l.o.g., we assume that $\mathcal B$ is the smaller set.

\begin{theorem}\label{thm:fptb}

An optimal solution of \hvrbs can be computed in $O(n\log n + n|\mathcal B|9^{|\mathcal B|})$ time.

\end{theorem}

We first give a
high-level description of the algorithm. It begins by subdividing 
the plane into at most $|\mathcal B|+1$ vertical strips, each consisting of the area
``between'' two horizontally successive blue points, and at most $|\mathcal B|+1$
horizontal strips, each consisting of the area ``between'' two vertically
successive blue points (see Figure \ref{fig:empty-corner}). Since each strip
can contain only red points in its interior, an optimal solution uses at most two lines inside
a single strip (Lemma~\ref{lem:atmost2_and_easy}(a)). We can therefore guess (by exhaustive
enumeration) the number of lines used in each strip in an optimal solution.
This gives a running time of roughly $9^{|\mathcal B|}$. A second observation
is that if an optimal solution uses two lines in a strip, these can be placed
as far away from each other as possible (Lemma~\ref{lem:atmost2_and_easy}(b)).
To complete the solution we must decide where to place the lines in
strips that contain only one line of an optimal solution. We consider every
pair of blue and red points whose separation may depend on the exact placement
of these lines. The key idea is that the separation of two such points can be
expressed as a \textsc{2-CNF} constraint.
If the upcoming formal exposition seems a bit more complicated than this informal idea, it is because we have to deal with points sharing the same x- or y-coordinates.

  \begin{figure}[t!]
  \centering
\begin{minipage}{0.55\textwidth}
  \centering
  \begin{tikzpicture}[scale=0.8]
    \def\xb{-1}
    \def\xe{6}
    \def\yb{-1}
    \def\ye{6}

    \node at (5.25,-0.25) {$p_b$} ;
    
    \foreach \i/\j in {0/1,1/3.6,2/3,5/0,2.5/4.4,3.6/5.5}{
      \draw (\xb,\j) -- (\xe,\j) ;
      \draw (\i,\yb) -- (\i,\ye) ;
      \node[bp] at (\i,\j) {} ;
    }

    \foreach \i in {0.4,1.5,3.1}{
      \draw[dashed] (\i,\yb) -- (\i,\ye) ;
    }
    \foreach \j in {2,4,4.5,5.4}{
      \draw[dashed] (\xb,\j) -- (\ye,\j) ;
    }

    \node[rp] at (3.2,1.9) {} ;
    \node at (3.4,1.7) {$p$} ;

    \draw[fill,opacity=0.1,very thick] (2.5,1) -- (2.5,3) -- (3.6,3) -- (3.6,1) -- cycle ;

    \draw[->] (3.9,1.95) -- (3.9,1.7) ;
    \draw[->] (3.15,1.2) -- (3.40,1.2) ;
    
  \end{tikzpicture}
  \subcaption{The cell decomposition (solid lines), a guess of how $S$ intersects it (dashed lines), and an interesting cell (in gray) for a point $p_b$ (bottom-right corner).
    The red point $p$ cannot be in the south-east quadrant of this cell which translates to the $2$-clause $y^2_p \lor \neg x^4_p$.
    Indeed, it should be that the horizontal line of $S$ is below it or that the vertical line is to its right.}
  \label{fig:empty-corner}
\end{minipage}
\qquad
\begin{minipage}{0.38\textwidth}
    \centering
    \begin{tikzpicture}[scale=0.6]
      \def\xb{-0.5}
      \def\xe{4}
      \def\yb{-0.3}
      \def\ye{0.7}
      
      \node[rp] at (0,0) {} ;
      \node[rp] at (3.2,0.4) {} ;

      \node at (0.25,0) {$p$} ;
      \node at (3.5,0.4) {$p'$} ;

      \fill[opacity=0.1] (\xb,\yb) -- (\xb,\ye) -- (\xe,\ye) -- (\xe,\yb) -- cycle ;

      \draw[dashed] (\xb,-0.2) -- (\xe,-0.2) ;

  \end{tikzpicture}
    \subcaption{Two consecutive red points in a horizontal strip $\mathcal R_h(i)$. If the corresponding line of $S$ is below $p$, then it is also below $p'$ which translates to $y^i_p \rightarrow y^i_{p'}$.}
    \label{fig:two-clause-consistency-h}

\qquad
    
        \centering
        \begin{tikzpicture}[scale=0.6,rotate=270]
          \def\xb{-0.5}
          \def\xe{4}
          \def\yb{-0.3}
          \def\ye{0.7}
      
      \node[rp] at (0,0) {} ;
      \node[rp] at (3.2,0.4) {} ;

      \fill[opacity=0.1] (\xb,\yb) -- (\xb,\ye) -- (\xe,\ye) -- (\xe,\yb) -- cycle ;

      \draw[dashed] (\xb,-0.2) -- (\xe,-0.2) ;

      \node at (0.25,0) {$p$} ;
      \node at (3.5,0.4) {$p'$} ;
      
  \end{tikzpicture}
    \subcaption{Two consecutive red points in a vertical strip $\mathcal R_v(j)$. If the corresponding line of $S$ is to the left of $p$, then it is also to the left of $p'$ which translates to $x^i_p \rightarrow x^i_{p'}$.}
    \label{fig:two-clause-consistency-v}
\end{minipage}
\caption{Illustration of the algorithm and the two kinds of clauses of the \textsc{2-SAT} instance.}
\label{fig:different-two-clauses}
\end{figure}
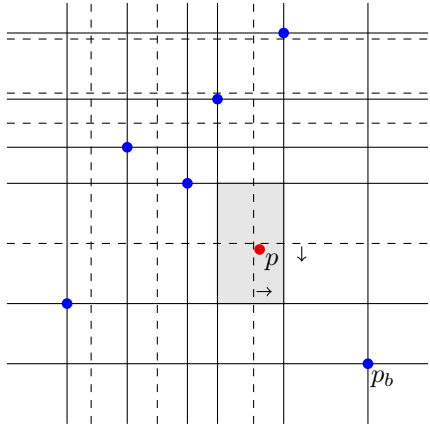
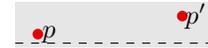
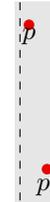

We now proceed to a formal description of our algorithm, beginning  with some definitions.
For a point $p\in\mathbb R^2$, let $p(x)$ and $p(y)$ be its $x$-coordinate and $y$-coordinate, respectively.
Also, let $X, Y$ be the sets of $x$, $y$ coordinates of the points in $\mathcal B$. In order to ease presentation later on, with a slight terminology abuse, we add $-\infty, +\infty$ to both $X$ and $Y$. Let $X(i), Y(i)$ be the respective $i$-th elements of these sets in increasing order with $0 \leqslant i$, and let $k = |X|-2$ and $l = |Y|-2$; $k\leqslant |B|$ and $l\leqslant |B|$.

\begin{definition}
The vertical strips are defined as $V_i = \{p\in\mathbb R^2 \mid  X(i) \leqslant p(x) \leqslant X(i+1)\}$ for $i\in [0, k]$. 
\end{definition}
\begin{definition}
The horizontal strips are defined as $H_i = \{p\in\mathbb R^2 \mid  Y(i) \leqslant p(y) \leqslant Y(i+1)\}$ for $i\in [0, l]$. 
\end{definition}

The horizontal and vertical strips defined above essentially partition the
plane into open monochromatic (red) or empty regions, while the boundaries of the strips may contain both red and blue points.
As a result, we have the following properties of an optimal solution.

\begin{lemma}\label[lemma]{lem:atmost2_and_easy}
(a) An optimal solution of \hvrbs contains at most two lines in each horizontal or vertical strip. (b) In the case where a strip has two lines, these lines can be assumed to be placed in a way such that all red points in the interior of the strip lie between them.
\end{lemma}

\lv{
\begin{proof}
Recall that our notion of separation forbids lines from passing through input points.
As the interior of every strip contains only red points, in any solution, every line that is between two other lines within a strip can be safely removed without affecting feasibility. Moreover, two lines within a strip can be translated in opposite directions towards the boundaries of the strip such that they enclose between them all red points that lie in the interior of the strip but no blue point. 
\end{proof}
}






We are now ready to give the proof of the main theorem of this section.

\begin{proof}[Proof of Theorem \ref{thm:fptb}]

We describe an FPT algorithm which guesses how many lines an optimal solution
uses in each strip and then produces a \textsc{2-SAT} instance of size
$O(|\mathcal B|n)$ in order to check if its guess is feasible. We assume that
we have access to two lists containing the input points sorted lexicographically by their $(x, y)$ and $(y, x)$ coordinates; producing these lists takes $O(n\log n)$ time.

Let $S$ be some optimal solution. 
We first guess how many lines of $S$ are in each horizontal and each vertical strip.
Since, by Lemma \ref{lem:atmost2_and_easy}, $S$ contains at most two lines per strip, and there are $l+1 \leqslant |\mathcal B|+1$ horizontal strips and $k+1 \leqslant |\mathcal B|+1$ vertical strips, there are at most $3^{|\mathcal B|+1}$ possibilities to guess from for each direction thus, $O(9^{|\mathcal B|})$ in total.

In what follows, we assume that we have fixed how many lines of $S$ are in each strip.
We describe an algorithm deciding in polynomial time if such a specification gives a feasible solution.  Since a specification fully determines the number of lines of
a solution, the algorithm simply goes through all specifications and selects one with minimum cost among all feasible ones.

We now produce a \textsc{2-SAT} instance which will be satisfiable if and only
if a given specification is feasible. We first define the variables: 
for each horizontal strip $H_i$ that contains exactly one line from $S$ and for each red point $p\in H_i$, 
we define a variable $y^i_p$. Its informal meaning is ``the line of $S$ in $H_i$
is below point $p$''. Note that when $p$ lies on the upper (lower) boundary of $H_i$, $y^i_p$ is set to true (false) by default.  Similarly, for each vertical strip $V_j$ that contains exactly one line from $S$ and for each red point $p\in V_j$, we define a variable $x^j_p$.  Its informal meaning is ``the line of $S$ in $V_j$ is to the left of $p$''. It is set to true (false) by default when $p$ lies on the right (left) boundary of $V_j$. We have constructed $O(n)$
variables (at most four for each point of $\mathcal R$).

Next, we construct \textsc{2-CNF} clauses imposing the informal meaning
described. For each strip $H_i$ that contains exactly one horizontal line from $S$ and each pair of red points $p, p'\in H_i$ that are consecutive in lexicographic $(y, x)$ order, we add the clause $(y^i_p \to y^i_{p'})$. Note that we can skip pairs that have a point lying on the upper or lower boundary of $H_i$ as the corresponding variable has been already set to true or false respectively and the clause is satisfied; see the description in the previous paragraph. Similarly, for each strip $V_j$ that contains exactly one vertical line from $S$ and each pair of red points $p, p'\in V_j$ that are consecutive in lexicographic $(x, y)$ order, we add the clause $(x^j_p \to x^j_{p'})$; as before, pairs that have a point lying on the left or right boundary of $V_j$ do not produce any clauses. Observe that given any solution, we can construct from its lines an assignment following the informal meaning described above that satisfies all clauses added so far, while from any satisfying assignment we can find lines according to the informal meaning.  We call the $O(n)$ clauses constructed so far the coherence part of our instance.


What remains is to add some further clauses to our instance to ensure not only
that each satisfying assignment encodes a solution, but also that the solution
is feasible, that is, it separates all pairs of red and blue points. 

Consider a cell $C_{ij} = H_i\cap V_j$, where $i\in[0,l]$ and $j\in[0,k]$. 
A red point $p\in C_{ij}$ is called $C_{ij}$-\emph{separable} for a point $p_b\in \mathcal B$, if $p$ can be separated from $p_b$ by a vertical or horizontal line running through the interior of $C_{ij}$.
We will sometimes call $p$ just separable when $C_{ij}$ and $p_b$ are obvious from the context. 
We say that $C_{ij}$ is \emph{interesting} for a point $p_b\in \mathcal B$ if the
following conditions hold: (i) $C_{ij}$ contains at least one red point that is $C_{ij}$-separable for $p_b$; (ii) at least one of $H_i$ or $V_j$ contains at most one horizontal or one vertical line from $S$ respectively; (iii) if $X(j+1) < p_b(x)$ or $p_b(x) < X(j)$, then there is no vertical line from $S$ in a strip between $p_b$ and $V_j$; and (iv)  if $Y(i+1) < p_b(y)$ or $p_b(y) < Y(i)$, then there is no horizontal line from $S$ in a strip between $p_b$ and $H_i$. Note that even if $C_{ij}$ is interesting for $p_b$, it may contain a red point $p$ that is \emph{already separated} from $p_b$ by a line going through $C_{ij}$: this happens exactly when $H_i$ or $V_j$ contains two horizontal or vertical lines from $S$ respectively and $p$ lies either in the interior of $C_{ij}$ or on its boundary but not on the same side of $H_i$ or $V_j$ as $p_b$. 

The motivation behind these definitions is that the cells that are interesting for $p_b$ contain exactly the red points that need to be separated from $p_b$ by lines going through the cells and whose positions cannot be predetermined. We therefore have to add some clauses to express these constraints. 


For each $p_b\in \mathcal B$ and each cell $C_{ij}$ that is interesting for $p_b$ we construct a clause for every
red point $p\in C_{ij}$ that is separable and not already separated from $p_b$. Initially, the clause is empty. If the specification says that there is exactly one line from $S$ in $H_i$, we add to the clause a literal as follows: if $y(p_b) \geqslant Y(i+1)$, we add $\neg y^{i}_p$ (meaning that the horizontal line is above $p$, and hence separates $p$ from $p_b$); if $y(p_b) \leqslant Y(i)$, we add $y^{i}_p$. Furthermore, if the specification says that there is exactly one line from $S$ in $V_j$, we add to the clause a literal as follows: if $x(p_b) \geqslant X(i+1)$, we add the literal $\neg x^{j}_p$; if $x(p_b) \leqslant X(i)$, we add $x^{j}_p$.
Observe that this process produces clauses of size at
most two.  It may produce an empty clause, rendering the \textsc{2-SAT}
unsatisfiable, in the case where there is no line of $S$ in $H_i$ or $V_j$, but this is desirable since in this case no feasible solution
matches the specification. Note that we have constructed $O(|\mathcal
B||\mathcal R|)$ clauses in this way (at most four for each pair of a blue with
a red point). Hence, the \textsc{2-SAT} formula we have constructed has $O(n)$
variables and $O(|\mathcal B|n)$ clauses. Since \textsc{2-SAT} can be solved in
linear time, we obtain the promised running time.

To complete the proof we rely on the informal correspondence between
assignments to the \textsc{2-SAT} instance and \hvrbs solutions. In particular,
if there exists a solution that agrees with the guessed specification, this
solution can easily be translated to an assignment that satisfies the coherence
part of the \textsc{2-SAT} formula. Furthermore, for any blue point $p_b$ and
any separable and not already separated red point $p$ in a cell $C_{ij}$ that is interesting for $p_b$, the
solution must be placing at least one line going through $C_{ij}$
in a way that separates $p_b$ from $p$ (this follows from the
fact that the cell is interesting).  Hence, the corresponding \textsc{2-SAT}
clauses are also satisfied. Conversely, given an assignment to
the \textsc{2-SAT} instance, we construct an \hvrbs solution following the
informal meaning of the variables. We first note that for every blue point $p_b$, every red point is $C_{ij}$-separable for $p_b$ for at least one cell $C_{ij}$. Observe that for any cell $C_{ij}$ that is not interesting for $p_b$ and contains at least one separable point, we have that either all red points in the cell are separated from $p_b$ by lines outside the cell or all separable red points in the cell are separated from $p_b$ by the four lines running through the cell. Furthermore, if $C_{ij}$ is interesting for $p_b$, then all separable (and not already separated) red points in the cell are separated from $p_b$ because of the additional \textsc{2-SAT} clauses we
added in the second part of the construction.  \end{proof}

\section{Open problems}
The most intriguing open problem is settling the complexity of \hvrbs w.r.t. the number of lines. We conjecture it to be FPT. Other problems include the complexity of \rbs when the lines can have three different slopes and of \hvrbs in $3$-dimensions. 

\medskip

\noindent
\textbf{Acknowledgements.} The authors would like to thank Sergio Cabello and Christian Knauer for fruitful discussions.

\bibliographystyle{abbrv}

\sv{
  \newpage
  \section{Appendix}

  \subsection{Some more details on \shs}

  \subsection{Why the lower bound might be elusive}

  Let us say a few words about the difficulty of showing such a result for \rbs, compared to its NP-hardness.
A set of $k$ lines creates at most $h(k) := {k+1 \choose 2}+1$ cells.
Therefore, any YES-instance can be covered by $h(k)$ pairwise-disjoint monochromatic convex sets.
This prevents us from encoding an adjacency matrix on $n$ vertices with bichromatic gadgets, while one does not seem to achieve much with a monochromatic encoding.

A perhaps more concrete issue with encoding an adjacency matrix is the following.
Suppose we try to reduce directly from \textsc{Multicolored Subgraph Isomorphism} (or its special case \textsc{Multicolored Clique}), and we want a horizontal line $L(u)$ to represent the choice of a vertex $u$ within one set, a vertical line $L(v)$ to represent the choice of a vertex $v$ in another set, and the lines are compatible iff $uv$ is an edge.
Here is the pitfall: if $uv$ and $uw$ are edges, then $L(u)$ should form a feasible solution with $L(v)$ and with $L(w)$; but then, it can be observed that every vertical line in between $L(v)$ and $L(w)$ also completes $L(u)$ into a feasible solution (which is undesirable as soon as there are vertices between $v$ and $w$ which are not adjacent to $u$).

We overcome those issues by reducing from \shs.

  \subsection{Correctness of the reduction}

  \subsection{Proof of \cref{lem:atmost2}}

  Suppose for contradiction that there are three lines in $S$ which have exactly
$i$ points to the left of them: they are the lines $x=x_1$, $x=x_2$, and
$x=x_3$, with $x_1<x_2<x_3$. Then all points with $x$-coordinates between
$x_1,x_3$ are red. Hence, removing the line $x=x_2$ from the solution produces
a better, but still feasible, solution. The argument is identical for
horizontal lines.

  \subsection{Proof of \cref{lem:easy2}}

  We state the proof for horizontal lines, since the case of vertical lines is
identical. Suppose that $S$ contains the lines $y=y_1$, $y=y_2$ which both have
exactly $i$ points of $\mathcal B$ below them, and $y_1<y_2$. Let $y_1'$ be
such that the line $y=y_1'$ contains exactly $i$ blue points below it and as
many red points above it a possible. Let $y_2'$ be such that the line $y=y_2'$
contains exactly $i$ blue points below it and as many red points below it as
possible. We replace in the solution the lines $y=y_1$ and $y=y_2$ with the
lines $y=y_1'$ and $y=y_2'$. Informally, we have moved the line $y=y_1$ as far
down as possible and the line $y=y_2$ as far up as possible without changing
the number of blue points each has below it. We claim that the new solution is
still feasible because any cell contained between $y=y_1'$ and $y=y_2'$
contains only red points, while all other cells are either unchanged or have
become smaller.  

}

\end{document}